\documentclass[11pt]{amsart}
\usepackage{amsthm}
\usepackage{amsmath,amstext,amsthm,amsfonts,amssymb}
\usepackage{multicol}
\usepackage{latexsym}
\usepackage{color}
\usepackage{graphicx}
\graphicspath{ {images/} }
\usepackage{epsf}
\usepackage{epsfig}
\usepackage{subfigure}
\usepackage{rotating}
\usepackage{curves}
\usepackage{epic}

\usepackage{graphics}
\usepackage{verbatim}
\usepackage{color}
\usepackage{hyperref}

\newtheorem{theorem}{Theorem}[section]
\newtheorem{lemma}[theorem]{Lemma}
\theoremstyle{definition}

\newtheorem{proposition}[theorem]{Proposition}

\theoremstyle{remark}

\numberwithin{equation}{section}
\newcommand{\abs}[1]{\lvert#1\rvert}

\newcommand{\C}{\mathbb{C}}

\usepackage{amssymb}
\usepackage{amstext}
\usepackage{amsmath}
\usepackage{latexsym}
\usepackage{color}

\newcommand{\quat}{\mathbb H}
\newcommand{\R}{\Bbb R}
\newcommand{\Z}{\mathbb Z}

\newcommand{\be}{\begin{equation}}
\newcommand{\en}{\end{equation}}

\newcommand{\bedefin}{\begin{defi}}
\newcommand{\findefi}{\end{defi} \medskip}

\newcommand{\betheo}{\begin{theorem}$\!\!${\bf \,\,\,}}
\newcommand{\entheo}{\end{theorem}}
\newcommand{\enth}{\end{theorem}}

\newcommand{\becor}{\begin{cor}$\!\!${\bf .}}
\newcommand{\encor}{\end{cor}}

\newcommand{\belem}{\begin{lem}$\!\!${\bf .}}
\newcommand{\enlem}{\end{lem}}

\newcommand{\bea}{\begin{eqnarray}}
\newcommand{\ena}{\end{eqnarray}}

\newcommand{\beano}{\begin{eqnarray*}}
\newcommand{\enano}{\end{eqnarray*}}

\newcommand{\bee}{\begin{enumerate}}
\newcommand{\ene}{\end{enumerate}}

\newcommand{\bei}{\begin{itemize}}
\newcommand{\eni}{\end{itemize}}

\newcommand{\betab}{\begin{tabular}}
\newcommand{\entab}{\end{tabular}}

\newcommand{\bd}{\begin{displaymath}}

\newcommand{\h}{{\mathfrak H}}
\newcommand{\hquat}{\mbox{\boldmath $\mathfrak K_{\mathbb H}$}}
\newcommand{\kc}{{\mathfrak K_{\mathbb C}}}

\newcommand{\bsigma}{\mbox{\boldmath $\sigma$}}

\newcommand{\bPhi}{\mbox{\boldmath $\Phi$}}

\newcommand{\bPsi}{\mbox{\boldmath $\Psi$}}

\newcommand{\bbra}{\mbox{\boldmath $($}}
\newcommand{\bket}{\mbox{\boldmath $)$}}
\newcommand{\bmid}{\mbox{\boldmath $\;\mid\;$}}

\newcommand{\bfrakf}{\mbox{\boldmath $\mathfrak f$}}

\newcommand{\bfrakq}{\mbox{\boldmath $\mathfrak q$}}

\newcommand{\bfrakx}{\mbox{\boldmath $\mathfrak x$}}
\newcommand{\bfraka}{\mbox{\boldmath $\mathfrak a$}}
\newcommand{\bfrakb}{\mbox{\boldmath $\mathfrak b$}}
\newcommand{\bfrako}{\mbox{\boldmath $\mathfrak o$}}
\newcommand{\bfrakk}{\mbox{\boldmath $\mathfrak k$}}

\newcommand{\bA}{\mathbf A}

\newcommand{\bb}{\mathbf b}

\newcommand{\bk}{\mathbf k}

\newcommand{\bq}{\mathbf q}

\newcommand{\bx}{\mathbf x}

\newcommand{\bi}{\mathbf i}
\newcommand{\bj}{\mathbf j}
\newcommand{\rh}{\mathfrak{H}_{\mathbb{R}}}
\newcommand{\ch}{\mathfrak{H}_{\mathbb{C}}}
\newcommand{\qh}{\mathfrak{H}_{\mathbb{H}}}
\newcommand{\qu}{\mathbf{q}}

\usepackage{subfigure}

\begin{document}

\title[Quaternionic wavelets ]{Discretization of quaternionic continuous wavelet transforms}
\author{A. Askari Hemmat$^{1}$}
\address{$^{1}$ Department of Applied Mathematics, Faculty of Mathematics and Computer, Shahid Bahonar University of Kerman, Kerman, Iran. }
\author{K. Thirulogasanthar$^{2}$}
\address{$^{2}$ Department of Computer Science and Software
Engineering, Concordia University, 1455 De Maisonneuve Blvd. West,
Montr\'eal, Qu\'ebec, H3G 1M8, Canada. }
\author{A. Krzyzak$^{2}$}
\email{askari@mail.uk.ac.ir}
\email{santhar@gmail.com}
\email{krzyzak@cs.concordia.ca }

\thanks{ The research of AK was partly supported by the Natural Sciences and Engineering Research Council of Canada (NSERC). The research of KT was partly supported by The Fonds de recherche du Qu\'ebec - Nature et technologies (FRQNT)}
\date{\today}
\dedicatory{Dedicated to S. Twareque Ali.}
\begin{abstract}
 A scheme to form a basis and a frame for a Hilbert space of quaternion valued square integrable function   from a basis and a frame, respectively, of a Hilbert space of complex valued square integrable functions is introduced. Using the discretization techniques for 2D-continuous wavelet transform of the $SIM(2)$ group, the quaternionic continuous wavelet transform, living in a complex valued Hilbert space of square integrable functions, of the quaternion wavelet group is discretized, and thereby, a discrete frame for quaternion valued Hilbert space of square integrable functions is obtained.
\end{abstract}
\maketitle
\pagestyle{myheadings}
%%%%%%%%%%%%%%%%%%%%%%%%%%%
%%%%%%%%%%%%%%%%%%%%%%%%%%%%%%%%%%%%%%%%%%%%%%%%%%%%%%%%%%%%%%%%%%%%%%
\section{Introduction}
Frames and wavelets play important role in many  physics and engineering. The continuous wavelet transform (CWT),  used extensively in signal analysis and image processing, is a joint time-frequency transform. This is in sharp contrast to the Fourier transform, which can be used either to analyze the frequency content of a signal, or its time profile, but not both at the same time. In the real case, the CWT is built out of the coherent states obtained from a unitary, irreducible and square-integrable representation of the one-dimensional affine group, $\R\rtimes\R^*$, a group of translations and dilations of the real line. For the complex case, exactly as in the one dimensional situation, wavelets can be derived from the complex affine group, $\C\rtimes\C^*$, which is also known as the similitude group of $\R^2$, denoted $SIM(2)$, consisting of dilations, rotations and translations of the plane. From the point of view of applications, these wavelets have become standard tools, for example, in image processing including radar imaging. For details see \cite{Van} and the many references therein. In a recent paper, \cite{AT}, a CWT from the quaternionic affine group, $\quat\rtimes\quat^*$, on a complex and on a quaternionic Hilbert space has been studied. A continuous wavelet transform on a quaternionic Hilbert space involves two complex functions. This means, compared to a wavelet transform on a complex Hilbert space, we get here two transforms. Moreover, the quaternionic affine group, which can also be called the dihedral similitude group of $\mathbb{R}^4$, is similar, though not quite, to having two copies of $SIM(2)$. In \cite{AT} we suggested that such a transform could be useful in studying stereophonic or stereoscopic signals.\\

In practice, for computational purposes, one uses a discretized version of the transform. Real life data cannot be handled with infinite information, either by human mind or by computers. In an infinite dimensional Hilbert space, even a discretized version of a continuous wavelet contains infinite amount of information about a signal. Any device invented by humans cannot transmit infinite amount of data in a finite time. Even in transmitting a finite amount of data there will be lost components (highly possible: nodes in transmitting wires or devices can cause this). In this regard, the reconstruction of a signal needs to be done with finite information in a reasonable time.  The whole procedure depends on how fast the reconstruction series converges and what kind of components of the signal can be lost and which part cannot be lost (in the way of transmitting). This is where choosing the right frame for the right problem enters. There is no universal frame that fits to every problem at hand.  As technology advances we shall face new problems every day and our search for finding a solution to solve them will continue as well. As we have suggested earlier, the discretization procedure, and thereby the discrete frame, we present in this manuscript is expected to help in resolving problems with stereophonic or stereoscopic signals.\\

The novelty of this manuscript may be put in words as follows: The 1-D continuous wavelet transform and its discretization grid (translation and dilation) cannot be directly transformed to 2D-case, because one needs to add the rotation parameter.  However, in the engineering literature, there were several attempts to generalize the 1-D case to 2-D through tensor products, which obviously exclude the rotation factor. It can be partially successful depending on what we are looking for.  The first attempt in considering the rotation factor, in solving this issue, was given in \cite {Van} for the $SIM (2) $ (it resulted from the Ph.D. thesis of one of the authors of \cite {Van}). Even in their construction there is a drawback, as the discretization grid keep expanding along the angular variable. In this manuscript we fix this issue and use it to discretize the quaternionic continuous wavelet transform. Along the way, we describe how we can obtain frames and bases for quaternionic square integrable function spaces using their complex or real counterparts.\\

The article is organized as follows: In section 2 we gather relevant materials about quaternions and quaternionic Hilbert spaces as needed to the rest of the manuscript. In Section 3 we address the issue concerning lifting a basis from a complex-valued Hilbert space to a quaternion-valued Hilbert space and the same concept for frames is presented in section 4. The 2D continuous wavelet transform is discretized in \cite{Van} with a discretization grid. A modified discretization grid and its advantages and the drawbacks of the grid presented in \cite{Van} are presented in section 5. In section 6, using the grid of section 5 continuous quaternionic wavelet transform on a complex-valued Hilbert space is discretized, and thereby, using the results of section 4 a wavelet frame on quaternion-valued Hilbert space is obtained.  Section 7 ends the article with a conclusion.

%%%%%%%%%%%%%%%%%%%%%%%%%%%%%%%%%%%%%%%%%%%%%%%%%%%%%%%%%%%%%%%%%%%%%%%%%%
\section{Mathematical preliminaries}
In order to make the paper self-contained, we recall a few facts about quaternions which may not be well-known. In particular, we revisit the $2\times 2$ complex matrix representations of quaternions and quaternionic Hilbert spaces. For details one could refer \cite{Ad, Gra1, colsabstru, ghimorper}.

\subsection{Quaternions}
Let $\quat$ denote the field of all quaternions and $\quat^*$ the group (under quaternionic
multiplication) of all
invertible quaternions. A general quaternion can be written as
$$\bfrakq = q_0 + q_1 \bi + q_2 \bj + q_3 \bk, \qquad q_0 , q_1, q_2, q_3 \in \mathbb R, $$
where $\bi,\bj,\bk$ are the three quaternionic imaginary units, satisfying
$\bi^2 = \bj^2 = \bk^2 = -1$ and $\bi\bj = \bk = -\bj\bi,  \; \bj\bk = \bi = -\bk\bj,
   \; \bk\bi = \bj = - \bi\bk$. The quaternionic conjugate of $\bfrakq$ is
$$ \overline{\bfrakq} = q_0 - \bi q_1 - \bj q_2 - \bk q_3 . $$
We shall use the
$2\times 2$ matrix representation of the quaternions, in which
$$
   \bi = \sqrt{-1}\sigma_1, \quad \bj = -\sqrt{-1}\sigma_2, \quad \bk  = \sqrt{-1}\sigma_3, $$
and the $\sigma$'s are the three Pauli matrices,
$$
  \sigma_1 = \begin{pmatrix} 0 & 1\\ 1& 0 \end{pmatrix}, \quad
\sigma_2 = \begin{pmatrix} 0 & -i\\ i & 0\end{pmatrix}, \quad
\sigma_3 = \begin{pmatrix} 1 & 0\\ 0 & -1 \end{pmatrix}, $$
to which we add
$$\sigma_0 = \mathbb I_2 = \begin{pmatrix} 1 & 0 \\ 0 & 1\end{pmatrix}.$$
We shall also use the matrix valued vector $\bsigma = (\sigma_1, -\sigma_2 , \sigma_3)$. Thus,
in this representation,
$$\bfrakq = q_0\sigma_0 + i\bq\cdot \bsigma =
\begin{pmatrix} q_0 + iq_3 & -q_2 +iq_1 \\ q_2 + iq_1 & q_0 -iq_3\end{pmatrix}, \quad
\bq = (q_1, q_2, q_3 ).  $$
In this representation, the quaternionic conjugate of $\bfrakq$ is given by $\bfrakq^\dag$.
Introducing  two complex variables, which we write as
$$ z_1 = q_0 + iq_3 , \qquad z_2 = q_2 + iq_1, $$
we may also write
\be
  \bfrakq = \begin{pmatrix} z_1 & -\overline{z}_2\\ z_2 & \overline{z}_1 \end{pmatrix}.
\label{comp-rep-quat}
\en
From this it is clear that the group $\quat^*$ is isomorphic to the {\em affine $SU(2)$ group},
i.e., $\mathbb R^{>0}\times SU(2)$, which is the group $SU(2)$ together with all (non-zero)
dilations. As a set  $\quat^* \simeq \mathbb R^{>0} \times S(4)$, where $S(4)$ is the surface of the sphere, or more simply, $\quat^* \simeq \mathbb R^4 \backslash \{\mathbf 0\}$.
From (\ref{comp-rep-quat}) we get
\be
  \text{det}[\bfrakq] = \vert z_1\vert^2 + \vert z_2\vert^2 = q_0^2 + q_1^2 + q_2^2 + q_3^2
     : = \vert \bfrakq \vert^2 ,
\label{quat-norm}
\en
$\vert \bfrakq \vert$ denoting the usual norm of the quaternion $\bfrakq$. Note also that
$$
  \bfrakq^\dag \bfrakq =   \bfrakq \bfrakq^\dag = \vert \bfrakq \vert^2\; \mathbb I_2 . $$
If $\bfrakq$ is invertible,
$$
  \bfrakq^{-1} =  \frac 1{\vert \bfrakq \vert^2 }
  \begin{pmatrix} \overline{z}_1  &  \overline{z}_2 \\ -z_2 & z_1 \end{pmatrix} .$$

%%%%%%%%%%%%%%%%%%%%%%%%%%%%%%%%%%%%%%%%%%%%%%%%%%%%%%%%%
\subsection{Quaternionic Hilbert spaces}
In this subsection we  define left and right quaternionic Hilbert spaces. For details we refer the reader to \cite{Ad, Gra1, colsabstru, ghimorper}. We also define the Hilbert space of square-integrable functions on quaternions following \cite{Vis}.
\subsubsection{Right Quaternionic Hilbert Space}
Let $V_{\quat}^{R}$ be a vector space under right multiplication by quaternions.  For $f,g,h\in V_{\quat}^{R}$ and $\bfrakq\in \quat$, the inner product
$$\langle\cdot\mid\cdot\rangle:V_{\quat}^{R}\times V_{\quat}^{R}\longrightarrow \quat$$
satisfies the following properties
\begin{enumerate}
\item[(i)]
$\overline{\langle f\mid g\rangle}=\langle g\mid f\rangle$
\item[(ii)]
$\|f\|^{2}=\langle f\mid f\rangle>0$ unless $f=0$, a real norm
\item[(iii)]
$\langle f\mid g+h\rangle=\langle f\mid g\rangle+\langle f\mid h\rangle$
\item[(iv)]
$\langle f\mid g\bfrakq\rangle=\langle f\mid g\rangle\bfrakq$
\item[(v)]
$\langle f\bfrakq\mid g\rangle=\overline{\bfrakq}\langle f\mid g\rangle$
\end{enumerate}
where $\overline{\bfrakq}$ stands for the quaternionic conjugate. It is always assumed that the
space $V_{\quat}^{R}$ is complete under the norm given above. Then,  together with $\langle\cdot\mid\cdot\rangle$ this defines a right quaternionic Hilbert space. Quaternionic Hilbert spaces share most of the standard properties of complex Hilbert spaces. The Dirac bra-ket notation
can be adapted to quaternionic Hilbert spaces:
$$\mid f\bfrakq\rangle=\mid f\rangle\bfrakq,\hspace{1cm}\langle f\bfrakq\mid=\overline{\bfrakq}\langle f\mid\;, $$
for a right quaternionic Hilbert space, with $\vert f\rangle$ denoting the vector $f$ and $\langle f\vert$ its dual vector.

%%%%%%%%%%%%%%%%%%%%%%%%%%%%%%%%%%%%%%%%%%%%%%%%%%%%%%%%%%%%%%
\subsubsection{Left Quaternionic Hilbert Space}
Let $V_{\quat}^{L}$ be a vector space under left multiplication by quaternions.  For $f,g,h\in V_{\quat}^{L}$ and $\bfrakq\in \quat$, the inner product
$$\langle\cdot\mid\cdot\rangle:V_{\quat}^{L}\times V_{\quat}^{L}\longrightarrow \quat$$
satisfies the following properties
\begin{enumerate}
\item[(i)]
$\overline{\langle f\mid g\rangle}=\langle g\mid f\rangle$
\item[(ii)]
$\|f\|^{2}=\langle f\mid f\rangle>0$ unless $f=0$, a real norm
\item[(iii)]
$\langle f\mid g+h\rangle=\langle f\mid g\rangle+\langle f\mid h\rangle$
\item[(iv)]
$\langle \bfrakq f\mid g\rangle=\bfrakq\langle f\mid g\rangle$
\item[(v)]
$\langle f\mid \bfrakq g\rangle=\langle f\mid g\rangle\overline{\bfrakq}$
\end{enumerate}
Again, we shall assume that the space $V_{\quat}^{L}$ together with $\langle\cdot\mid\cdot\rangle$ is a separable Hilbert space. Also,
\begin{equation}\label{leftcs}
\mid \bfrakq f\rangle=\mid f\rangle\overline{\bfrakq},\hspace{1cm}\langle \bfrakq f\mid=\bfrakq\langle f\mid.
\end{equation}
Note that, because of our convention for inner products, for a left quaternionic Hilbert space, the bra vector $\langle f\mid$ is to be identified with the vector itself, while the ket vector $\mid f \rangle$ is to be identified with its dual.
(There is a natural left multiplication by quaternionic scalars on the dual of a right quaternionic Hilbert space and a similar right multiplication on the dual of a left quaternionic Hilbert space.)

The field of quaternions $\quat$ itself can be turned into a left quaternionic Hilbert space by defining the inner product $\langle \bfrakq \mid \bfrakq^\prime \rangle = \bfrakq \bfrakq^{\prime\dag} = \bfrakq\overline{\bfrakq^\prime}$ or into a right quaternionic Hilbert space with  $\langle \qu \mid \qu^\prime \rangle = \bfrakq^\dag \bfrakq^\prime = \overline{\bfrakq}\bfrakq^\prime$.
%%%%%%%%%%%%%%%%%%%%%%%%%%%%%%%%%%%%%%%%%%%%%%%%%%%%%%%%%%%%%%%%
\subsubsection{Quaternionic Hilbert Spaces of Square-integrable Functions}
Let $(X, \mu)$ be a measure space and $\quat$  the field of quaternions, then
$$L^2_{\quat}(X,\mu)=\left\{f:X\rightarrow \quat\;\; \left| \;\; \int_X|f(x)|^2d\mu(x)<\infty \right.\right\}$$\label{L^2}
is a right quaternionic Hilbert space, with the (right) scalar product
\begin{equation}
\langle f \mid g\rangle =\int_X \overline{f(x)} g(x)\; d\mu(x),
\label{left-sc-prod}
\end{equation}
where $\overline{f(x)}$ is the quaternionic conjugate of $f(x)$, and (right)  scalar multiplication $f\bfraka , \; \bfraka\in \quat,$ with $(f \bfraka)(x) = f(x)\bfraka $ (see \cite{Vis} for details). Similarly, one could define a left quaternionic Hilbert space of square-integrable functions.

%%%%%%%%%%%%%%%%%%%%%%%%%%%%%%%%%%%%%%%%%%%%%%%%%%%%%%%%%%%%%%%%%%%%%%%%%%%%
\subsection{Some notations}
Let $X$ be a locally compact space with measure $\nu$. We form the following Hilbert spaces of square integrable functions.
\begin{eqnarray*}
\rh&=&L^2_{\R}(X, d\nu)=\left\{f:X\longrightarrow\R~\vert~\int_X\abs{f(x)}^2d\nu(x)<\infty\right\}
\end{eqnarray*}
and in the same way $\ch=L^2_{\C}(X, d\nu)$ and
$\qh=L^2_{\quat}(X,d\nu)$. When $X=\quat$, we denote these spaces by $\mathfrak{K}_{\R}, \kc, \hquat$ respectively.
%%%%%%%%%%%%%%%%%%%%%%%%%%%%%%%%%%%%%%%%%%%%%%%%%%%%%%%%%%%%%%%%%%%%%%%%%%%%%%%%%%%%%%%%
\subsection{A right quaternionic Hilbert space}\label{subsec-quat-hilb-sp}
We consider the Hilbert space $\hquat$, of quaternion-valued functions over the quaternions.
An element  $\bfrakf \in \hquat$ has the form
\be
 \bfrakf (\bfrakx) = \begin{pmatrix} f_1 (\bfrakx ) & -\overline{f_2 (\bfrakx )}\\
                                      f_2 (\bfrakx ) & \overline{f_1 (\bfrakx )}\end{pmatrix},
                                      \quad \bfrakx \in \mathbb H ,
\label{hquat-vect}
\en
where $f_1$ and $f_2$ are two complex-valued functions over the quaternions. The norm in $\hquat$ is given by
\be
  \Vert \bfrakf\Vert_{\mathfrak K_\quat}^2 = \int_{\mathbb H}\bfrakf (\bfrakx )^\dag \bfrakf (\bfrakx )
  \; d\bfrakx = \int_{\mathbb H}\vert\bfrakf (\bfrakx )\vert^2 \; d\bfrakx
      = \left[\int_{\quat} \left(\;\vert f_1 (\bfrakx ) \vert^2 + \vert f_2 (\bfrakx ) \vert^2\;\right)\; d\bfrakx \;
        \right]\sigma_0 ,
\label{bquat-norm}
\en
the finiteness of which implies that both $f_1$ and $f_2$ have to be elements of $\kc =
L^2_{\mathbb C}(\mathbb H , d\bfrakx )$, so that we may write
$$
   \Vert \bfrakf\Vert_{\mathfrak K_\quat}^2 = \left(\; \Vert f_1 \Vert_{\mathfrak K_\mathbb C}^2 +  \Vert f_2 \Vert_{\mathfrak K_\mathbb C}^2 \; \right)\sigma_0 . $$
In view of this, we may also write $\hquat = L^2_\mathbb H ( \quat , d\bfrakx )$.
In using the ``bra-ket'' notation we shall use the notation and convention:
\be
 \bbra \bfrakf \bmid =   \begin{pmatrix}\langle  f_1 \vert   & \langle f_2 \vert\\
                                     -\langle \overline{f}_2 \vert &  \langle \overline{f}_1 \vert\end{pmatrix},
 \quad \text{and} \quad
  \bmid \bfrakf \bket =   \begin{pmatrix}\vert f_1 \rangle   & -  \vert \overline{f}_2 \rangle\\
                                        \vert f_2 \rangle &  \vert \overline{f}_1 \rangle
                                        \end{pmatrix},
  \label{bra-ket}
 \en
 $\overline{f}$ denoting, as usual, the complex conjugate of the function $f.$
 The scalar product of two vectors $\bfrakf, \bfrakf' \in \hquat$ is
 \bea
  \bbra \bfrakf \bmid \bfrakf' \bket & = & \int_{\mathbb H}\bfrakf (\bfrakx )^\dag \bfrakf' (\bfrakx )
  \; d\bfrakx \nonumber\\
  & = & \begin{pmatrix}
  \langle f_1 \mid f_1^\prime\rangle_{\mathfrak H_\mathbb C} +
  \langle f_2 \mid f_2^\prime\rangle_{\mathfrak H_\mathbb C}
 & -\langle f_2^\prime \mid \overline{f}_1\rangle_{\mathfrak H_\mathbb C} +
 \langle f_1^\prime \mid \overline{f}_2\rangle_{\mathfrak H_\mathbb C} \\
  \langle \overline{f}_2^\prime \mid {f}_1\rangle_{\mathfrak H_\mathbb C} -
 \langle \overline{f}_1^\prime \mid {f}_2\rangle_{\mathfrak H_\mathbb C}
 & \langle f_1^\prime \mid f_1 \rangle_{\mathfrak H_\mathbb C} +
  \langle f_2^\prime \mid f_2\rangle_{\mathfrak H_\mathbb C} .
\end{pmatrix}
\label{hquat-sc-prod}
  \ena
  Note that
$$ \bbra \bfrakf \bmid \bfrakf' \bket^\dag =\bbra \bfrakf' \bmid \bfrakf \bket . $$
We see that if $\bfrakf$ is orthogonal to $\bfrakf'$ in $\hquat$ then
$$ \langle f_1 \mid f_1^\prime\rangle_{\mathfrak H_\mathbb C} +
  \langle f_2 \mid f_2^\prime\rangle_{\mathfrak H_\mathbb C} =
    \langle \overline{f}_2^\prime \mid {f}_1\rangle_{\mathfrak H_\mathbb C} -
 \langle \overline{f}_1^\prime \mid {f}_2\rangle_{\mathfrak H_\mathbb C}  = 0 .$$
In other words, defining the two vectors $\mathbf f = \begin{pmatrix} f_1 \\ f_2 \end{pmatrix} , \; \mathbf f' = \begin{pmatrix} f_1'\\ f_2'\end{pmatrix} \in
\kc \oplus \kc$, the orthogonality of  $\bfrakf$ and  $\bfrakf'$ in $\hquat$ implies
$\langle \mathbf f \mid \mathbf f' \rangle = 0$, i.e., the  orthogonality of   $\mathbf f$
and $\mathbf f'$
in $\kc \oplus \kc$ and in addition  (in an obvious notation), that  $\mathbf f \wedge \mathbf f' = 0$.

Multiplication by quaternions on $\hquat$ is defined from the right:
$$
  (\hquat \times \quat )\ni (\bfrakf , \bfrakq )\longmapsto \bfrakf \bfrakq, \quad \text{such that}
  \quad (\bfrakf \bfrakq) (\bfrakx ) = \bfrakf (\bfrakx )\bfrakq , $$
i.e., we take $\hquat$ to be a right quaternionic Hilbert space. This convention is consistent with the scalar product (\ref{hquat-sc-prod}) in the sense that
$$ \bbra \bfrakf \bmid \bfrakf'\bfrakq \bket = \bbra \bfrakf \bmid \bfrakf' \bket\bfrakq  \quad
\text{and} \quad \bbra \bfrakf\bfrakq \bmid \bfrakf' \bket =
\bfrakq^{\dagger}\bbra \bfrakf \bmid \bfrakf' \bket .$$

 On the other hand,
the action of operators $\bA$ on vectors $\bfrakf \in \hquat$ will be from the left $(\bA, \bfrakf)
\longmapsto \bA\bfrakf$. In particular, an operator $A$ on $\kc$ defines an operator $\bA$ on $\hquat$
as,
$$
   (\bA\bfrakf)(\bfrakx) = \begin{pmatrix} (Af_1) (\bfrakx ) & -\overline{(Af_2) (\bfrakx )}\\
                                      (Af_2) (\bfrakx ) & \overline{(Af_1) (\bfrakx )}\end{pmatrix} .$$
Multiplication of operators by quaternions will also be from the left. Thus, $\bfrakq\bA$ acts on
the vector $\bfrakf$ in the manner
$$ (\bfrakq\bA\bfrakf)(\bfrakx) = \bfrakq (\bA\bfrakf)(\bfrakx) . $$
We shall also need the``rank-one operator''
\bea\label{V1}
\bmid \bfrakf \bket\bbra \bfrakf' \bmid  & = &
\begin{pmatrix} \vert f_1 \rangle& - \vert \overline{f}_2 \rangle\\
                                        \vert f_2 \rangle &  \vert \overline{f}_1 \rangle
                                        \end{pmatrix}
 \begin{pmatrix}\langle  f_1' \vert  & \langle f_2' \vert\\
              -\langle \overline{f}_2' \vert &  \langle \overline{f}_1' \vert\end{pmatrix}
              \nonumber \\
    & = & \begin{pmatrix} \vert f_1 \rangle\langle f_1'\vert
                           + \vert \overline{f}_2  \rangle\langle \overline{f}_2'\vert
               & \vert f_1 \rangle\langle f_2'\vert
                           - \vert \overline{f}_2  \rangle\langle \overline{f}_1'\vert \\
                  -\vert \overline{f}_1 \rangle\langle \overline{f}_2'\vert
                           + \vert {f}_2  \rangle\langle {f}_1'\vert
               &  \vert \overline{f}_1  \rangle\langle \overline{f}_1'\vert
                          + \vert f_2 \rangle\langle f_2'\vert\; .
 \end{pmatrix}
 \ena
%%%%%%%%%%%%%%%%%%%%%%%%%%%%%%%%%%%%%%%%%%%%%%%%%%%%%%%%%%%%%%%%%%%%%%%
\section{Base lifting}
In this section we provide a general scheme for lifting a basis to quaternionic Hilbert spaces from their real and complex counterparts. However, one should be aware that the bases lifted here are only a possibility. In comparing the three spaces, real, complex and quaternion-valued square integrable function, the quaternionic Hilbert space is so gigantic. The procedure described below is just a lifting scheme from the reals and complexes to quaternions.\\
 It is known that if $\{\phi_n\}_{n=0}^{\infty}$ is a basis for $\rh$, then it is also a basis for $\ch$ in the sense that, for $f\in\rh$ and $g\in\ch$,
$$f=\sum_{n=0}^{\infty}a_n\phi_n;\quad a_n\in\R\quad\mbox{and}\quad g=\sum_{n=0}^{\infty}b_n\phi_n;\quad b_n\in\C.$$
\begin{lemma}\label{L0}
If $\{\phi_n\}_{n=0}^\infty$ is an orthonormal basis of $\ch$, then  $\{\overline{\phi}_n\}_{n=0}^\infty$ is also a basis for $\ch$.
\end{lemma}
An orthonormal basis in $\hquat$ can be  built using an orthonormal basis in $\kc$.
\begin{proposition}\label{P0}
Let
$\{\phi_n\}_{n=0}^\infty$ be an orthonormal basis of $\kc = L^2_{\mathbb C} (\quat, d\bfrakx )$. Define the
vectors
\be
  \bmid \bPhi_n\bket =  \frac 1{\sqrt{2}} \begin{pmatrix} \vert \phi_n \rangle&  \vert \phi_n \rangle\\
                                        -\vert \overline{\phi}_n \rangle &  \vert \overline{\phi}_n \rangle
                                        \end{pmatrix}, \quad n = 0,1,2, \ldots ,
\en
in $\hquat$. Then the family $\{\bPhi_n\}_{n=0}^{\infty}$ is a basis for $\hquat$.
\end{proposition}
\begin{proof}
 It is easy to check that these vectors are orthonormal in $\hquat$ and
$$ \bbra \bPhi_n \bmid \bfrakf\bket =
\frac 1{\sqrt{2}}\begin{pmatrix}
  \langle \phi_n \mid f_1\rangle_{\mathfrak H_\mathbb C} -
  \langle \overline{\phi}_n \mid f_2\rangle_{\mathfrak H_\mathbb C}
 & -\langle f_2 \mid \overline{\phi}_n\rangle_{\mathfrak H_\mathbb C} -
 \langle f_1 \mid \phi_n\rangle_{\mathfrak H_\mathbb C} \\
  \langle \overline{f}_2 \mid \phi_n\rangle_{\mathfrak H_\mathbb C} +
 \langle \overline{f}_1 \mid \overline{\phi}_n\rangle_{\mathfrak H_\mathbb C}
 & \langle f_1 \mid \phi_n \rangle_{\mathfrak H_\mathbb C} -
  \langle f_2 \mid \overline{\phi}_n\rangle_{\mathfrak H_\mathbb C}
\end{pmatrix} .$$
Now, $\bbra \bPhi_n \bmid \bfrakf\bket=0$ implies
$\langle f_1\vert \phi_n \rangle=0$ and $\langle f_2\vert \overline{\phi}_n \rangle=0$. From the fact that  $\{\phi_n\}_{n=0}^{\infty}$ and $\{\overline{\phi}_n\}_{n = 0}^\infty$ are orthonormal bases for $\ch$, we get $f_1=f_2=0$, and thereby,  $\bmid\bfrakf\bket=0$.
\end{proof}
Indeed, with
$$
      \bmid \bfrakf \bket = \begin{pmatrix} \vert f_1 \rangle& - \vert \overline{f}_2 \rangle\\
                                        \vert f_2 \rangle &  \vert \overline{f}_1 \rangle
                                        \end{pmatrix} \in  L^2_{\mathbb H} (\quat, d\bfrakx ), $$
and writing
$$
 \vert f_1 \rangle = \sum_{n=0}^\infty b_n \vert\phi_n\rangle, \;\;
 \vert f_2 \rangle = \sum_{n=0}^\infty c_n \vert\phi_n\rangle,  \quad \text{with} \quad
  b_n = \langle \phi_n \mid f_1\rangle ,\;\;  c_n = \langle \phi_n \mid f_2\rangle , $$
we easily verify that
$$  \bmid \bfrakf \bket = \sum_{n=0}^\infty\bmid \bPhi_n\bket \bfrakq_n , $$
where
$ \bfrakq_n =\bbra \bPhi_n \bmid \bfrakf\bket.$
A simpler basis can also be chosen for $\qh$ using the basis of $\ch$.
\begin{proposition}\label{P1}
Let $\{\phi_n\}_{n=0}^{\infty}$ be an orthonormal basis for $\ch$, then
\be
  \bmid \bPhi_n\bket = \begin{pmatrix} \vert \phi_n \rangle&  0\\
                                       0 &  \vert \overline{\phi}_n \rangle
                                        \end{pmatrix}, \quad n = 0,1,2, \ldots ,
\label{quat-onb1}
\en
is an orthonormal basis for $\qh$.
\end{proposition}
\begin{proof}
Let $\bmid\bfrakf\bket\in\qh$, then from (\ref{bra-ket}), we can see that
$$\bbra\bfrakf\bmid\bPhi_n\bket=
 \begin{pmatrix} \langle f_1\vert \phi_n \rangle&  \langle f_2\vert \overline{\phi}_n \rangle\\
                                        -\langle\overline{f}_2\vert \phi_n \rangle &  \langle\overline{f}_1\vert \overline{\phi}_n \rangle
                                        \end{pmatrix}, \quad n = 0,1,2, \ldots ,
$$
hence, if $\bbra\bfrakf\bmid\bPhi_n\bket=0$, then $\langle f_1\vert \phi_n \rangle=0$ and $\langle f_2\vert \overline{\phi}_n \rangle=0$. From the fact that  $\{\phi_n\}_{n=0}^{\infty}$ and $\{\overline{\phi}_n\}_{n = 0}^\infty$ are orthonormal bases for $\ch$, we get $f_1=f_2=0$, and thereby,  $\bmid\bfrakf\bket=0$.
\end{proof}

%%%%%%%%%%%%%%%%%%%%%%%%%%%%%%%%%%%%%%%%%%%%%%%%%%%%%%%%%%%%%%%%%%%%%%%%%%%%%%%%%%%%%%%%%%%%
\section{Frame lifting}
The argument given for the base lifting applies to frames as well. There is no need to elaborate any further.
\begin{lemma}\label{L1}
Let $\{\phi_n\}_{n=0}^{\infty}$ be a frame for $\kc$ with the lower frame bound $A$ and the upper frame bound $B$, then $\{\overline{\phi}_n\}_{n = 0}^\infty$ is also a frame for $\kc$ with the same frame bounds.
\end{lemma}

\begin{theorem}\label{FLT}
Let $\{\phi_n\}_{n=0}^{\infty}$ be a frame for $\kc$ with lower frame bound $A$ and upper frame bound $B$, then $\{\bPhi_n\}_{n = 0}^\infty$ and $\{\bPsi_n\}_{n = 0}^\infty$ are frames for $\hquat$ with the same frame bounds, where
\be
  \bmid \bPhi_n\bket =  \frac 1{\sqrt{2}} \begin{pmatrix} \vert \phi_n \rangle&  \vert \phi_n \rangle\\
                                        -\vert \overline{\phi}_n \rangle &  \vert \overline{\phi}_n \rangle
                                        \end{pmatrix},\quad
 \bmid \bPsi_n\bket =  \begin{pmatrix} \vert \phi_n \rangle& 0\\
                                       0 &  \vert \overline{\phi}_n \rangle
                                        \end{pmatrix},
 \quad n = 0,1,2, \ldots ,
\en
\end{theorem}
\begin{proof}
The dual vectors of $\bmid \bPhi_n\bket$ and $\bmid \bPsi_n\bket$, respectively, are
\be
  \bbra \bPhi_n\bmid =  \frac 1{\sqrt{2}} \begin{pmatrix} \langle \phi_n \vert&  -\langle \overline{\phi}_n \vert\\
                                        \langle \phi_n \vert &  \langle \overline{\phi}_n \vert
                                        \end{pmatrix},\quad
\bbra \bPsi_n\bmid =   \begin{pmatrix} \langle \phi_n \vert&  0\\
                                        0 &  \langle \overline{\phi}_n \vert
                                        \end{pmatrix},
 \quad n = 0,1,2, \ldots ,
\en
The projection operators are
\be\label{V2}
\bmid \bPsi_n\bket\bbra \bPsi_n\bmid=\bmid \bPhi_n\bket\bbra \bPhi_n\bmid=\begin{pmatrix} \vert\phi_n\rangle\langle \phi_n \vert& 0\\
                                        0 & \vert\overline{\phi}_n\rangle\langle \overline{\phi}_n \vert
                                        \end{pmatrix}, \quad n = 0,1,2, \ldots ,
\en
Since
\begin{eqnarray*}
\langle f_1\vert\phi_n\rangle&=&\langle\overline{\phi}_n\vert\overline{f}_1\rangle\\
\langle\phi_n\vert\overline{f}_2\rangle&=&\langle f_2\vert\overline{\phi}_n\rangle\\
\langle\overline{f}_1\vert\overline{\phi}_n\rangle&=&\langle\phi_n\vert f_1\rangle\\
\langle\overline{f}_2\vert\phi_n\rangle&=&\langle\overline{\phi}_n\vert f_2\rangle,
\end{eqnarray*}
we have
\begin{eqnarray*}
\langle f_1\vert \phi_n\rangle\langle\phi_n\vert f_1\rangle+\langle f_2\vert \overline{\phi}_n\rangle\langle\overline{\phi}_n\vert f_2\rangle&=&
\vert\langle f_1\vert \phi_n\rangle\vert^2+\vert\langle f_2\vert \overline{\phi}_n\rangle\vert^2\\
-\langle f_1\vert \phi_n\rangle\langle\phi_n\vert \overline{f}_2\rangle+\langle f_2\vert \overline{\phi}_n\rangle\langle\overline{\phi}_n\vert \overline{f}_1\rangle&=&0\\
-\langle \overline{f}_2\vert \phi_n\rangle\langle\phi_n\vert f_1\rangle+\langle \overline{f}_1\vert \overline{\phi}_n\rangle\langle\overline{\phi}_n\vert f_2\rangle&=&0\\
\langle \overline{f}_2\vert \phi_n\rangle\langle\phi_n\vert \overline{f}_2\rangle+\langle \overline{f}_1\vert \overline{\phi}_n\rangle\langle\overline{\phi}_n\vert \overline{f}_1\rangle&=&
\vert\langle \overline{f}_2\vert \phi_n\rangle\vert^2+\vert\langle\overline{f}_1\vert \overline{\phi}_n\rangle\vert^2
\end{eqnarray*}
Therefore, from (\ref{bra-ket}) and (\ref{V2}) we have, for $ n = 0,1,2, \ldots$,
\be\label{V3}
\bmid\bbra\bfrakf\bmid \bPsi_n\bket\bmid^2=\bmid\bbra\bfrakf\bmid \bPhi_n\bket\bmid^2=\begin{pmatrix}\vert\langle f_1\vert \phi_n\rangle\vert^2+\vert\langle f_2\vert \overline{\phi}_n\rangle\vert^2 & 0\\
                                        0 & \vert\langle \overline{f}_2\vert \phi_n\rangle\vert^2+\vert\langle\overline{f}_1\vert \overline{\phi}_n\rangle\vert^2
                                        \end{pmatrix}.
\en
Since the series in the frame condition converges, we have from lemma (\ref{L1}), 
\begin{eqnarray*}
 & &A(\|f_1\|^2_{\kc}+\|f_2\|^2_{\kc})\leq\sum_{n=0}^{\infty}(\vert\langle f_1\vert \phi_n\rangle\vert^2+\vert\langle f_2\vert \overline{\phi}_n\rangle\vert^2)\\
&=&\sum_{n=0}^{\infty}\vert\langle f_1\vert \phi_n\rangle\vert^2+\sum_{n=0}^{\infty}\vert\langle f_2\vert \overline{\phi}_n\rangle\vert^2\leq B(\|f_1\|^2_{\kc}+\|f_2\|^2_{\kc})
\end{eqnarray*}
Similarly, since $\|\overline{f}_1\|^2=\|f_1\|^2$ and $\|\overline{f}_2\|^2=\|f_2\|^2$, we have
\be
A(\|f_1\|^2_{\kc}+\|f_2\|^2_{\kc})\leq\sum_{n=0}^{\infty}(\vert\langle \overline{f}_2\vert \phi_n\rangle\vert^2+\vert\langle\overline{f}_1\vert \overline{\phi}_n\rangle\vert^2)\leq B(\|f_1\|^2_{\kc}+\|f_2\|^2_{\kc}).
\en
Therefore (matrix inequalities can be understood in the sense of positive definite matrices) we get
\be\label{FC}
A(\|f_1\|^2_{\kc}+\|f_2\|^2_{\kc})\sigma_0\leq\sum_{n=0}^{\infty}\bmid\bbra\bfrakf\bmid \bPhi_n\bket\bmid^2=\sum_{n=0}^{\infty}\bmid\bbra\bfrakf\bmid \bPsi_n\bket\bmid^2\leq B(\|f_1\|^2_{\kc}+\|f_2\|^2_{\kc})\sigma_0.
\en
Thus
\be
A\|\bfrakf\|^2_{\hquat}\leq\sum_{n=0}^{\infty}\bmid\bbra\bfrakf\bmid \bPhi_n\bket\bmid^2=\sum_{n=0}^{\infty}\bmid\bbra\bfrakf\bmid \bPsi_n\bket\bmid^2\leq B\|\bfrakf\|^2_{\hquat};\quad\forall\bfrakf\in\hquat.
\en
\end{proof}
%%%%%%%%%%%%%%%%%%%%%%%%%%%
%\subsection{Frame operators and snugness}
%\textcolor[rgb]{1,0,0}{We need to discuss this further}
%%%%%%%%%%%%%%%%%%%%%%%%%%%%%%%%%%%%%%%%%%%%%%%%%%%%%%%%%%%%%%%%%%
\section{2D continuous wavelet transform and its discretization}
\subsection{The continuous wavelet transform}
Let $\C^*=\{z\in\C~\vert~z\not=0\}$ and $G^{\C}_{\text{aff}}=\C\rtimes\C^*$ with  group operation
$$(z,w)(z',w')=(z+wz',ww').$$
Let $\underline{b}=(b_1,b_2)^T$,  $r_{\theta}=\left(\begin{array}{cc}\cos\theta&\sin\theta\\-\sin\theta&\cos\theta\end{array}\right)\in SO(2);~~\theta\in[0,2\pi)$ and $\lambda>0$, then $G^{\C}_{\text{aff}}$ can be considered as
$$G^{\C}_{\text{aff}}=SIM(2)=\left\{(\underline{b},\lambda,r_{\theta})~\vert~ \underline{b}\in\mathbb{R}^2, \lambda>0, r_{\theta}\in SO(2)\right\}.$$
The affine action can be written as
$$\underline{x}\mapsto \lambda r_{\theta}\underline{x}+\underline{b};\qquad\underline{x}\in\mathbb{R}^2\approxeq\C.$$
Note that $G^{\C}_{\text{aff}}$ can also be represented in matrix form:
$$G^{\C}_{\text{aff}}=\left\{g=\left(\begin{array}{cc}\lambda r_{\theta}&\underline{b}\\ \underline{0}^T&1\end{array}\right)~\vert~\underline{b}\in\mathbb{R}^2, \lambda>0, r_{\theta}\in SO(2)\right\}.$$
Now consider the group representation $U:G^{\C}_{\text{aff}}\longrightarrow L^2(\mathbb{R}^2)$ by
\begin{equation}\label{W1}
[U(\underline{b},\lambda,r_{\theta})\psi](\underline{x})
=\frac{1}{\lambda}\psi\left(\frac{r_{\theta}^{-1}(\underline{x}-\underline{b})}{\lambda}\right)
\end{equation}
\label{2dimrep}
Then, it is known that the representation (\ref{2dimrep}) is unitary, irreducible and square-integrable with respect to the measure $\displaystyle dg=a^{-3}d^2\overrightarrow{b}dad\theta $, for details we refer the reader to \cite{Van}.
%%%%%%%%%%%%%%%%%%%%%%
\subsection{ A refined discretization }
In \cite{Van} the continuous wavelet transform (\ref{2dimrep}) is discretized with the following discretization grid.
\begin{enumerate}
\item[$\bullet$] For dilation $a_j=\lambda^{-j},~~j\in\Z$ for some $\lambda>1$.
\item[$\bullet$]For the rotation the interval $[0,2\pi)$ is divided uniformly into $L_0$ pieces as $\theta_l=\frac{l2\pi}{L_0}$, where $l\in\Z_{L_0}=\{0,1,\cdots, L_0-1\}$.
\item[$\bullet$] For the translations it was taken as $\overrightarrow{b}_m=\overrightarrow{b}_{jlm_0m_1}=\lambda^{-j}rl\theta_0(m_o\beta_0,m_1\beta_1)$ with $m_0,m_1\in\Z,~\beta_0,\beta_1\geq 0, l\in\Z_{L_0}$.
\end{enumerate}
From the applications point of view, there is a  drawback with this grid. As the radial parameter moves away from the origin, the angular part gets out of control, that is, it will keep expanding, see Figure 1.
\begin{figure}[h]
\centering
\include{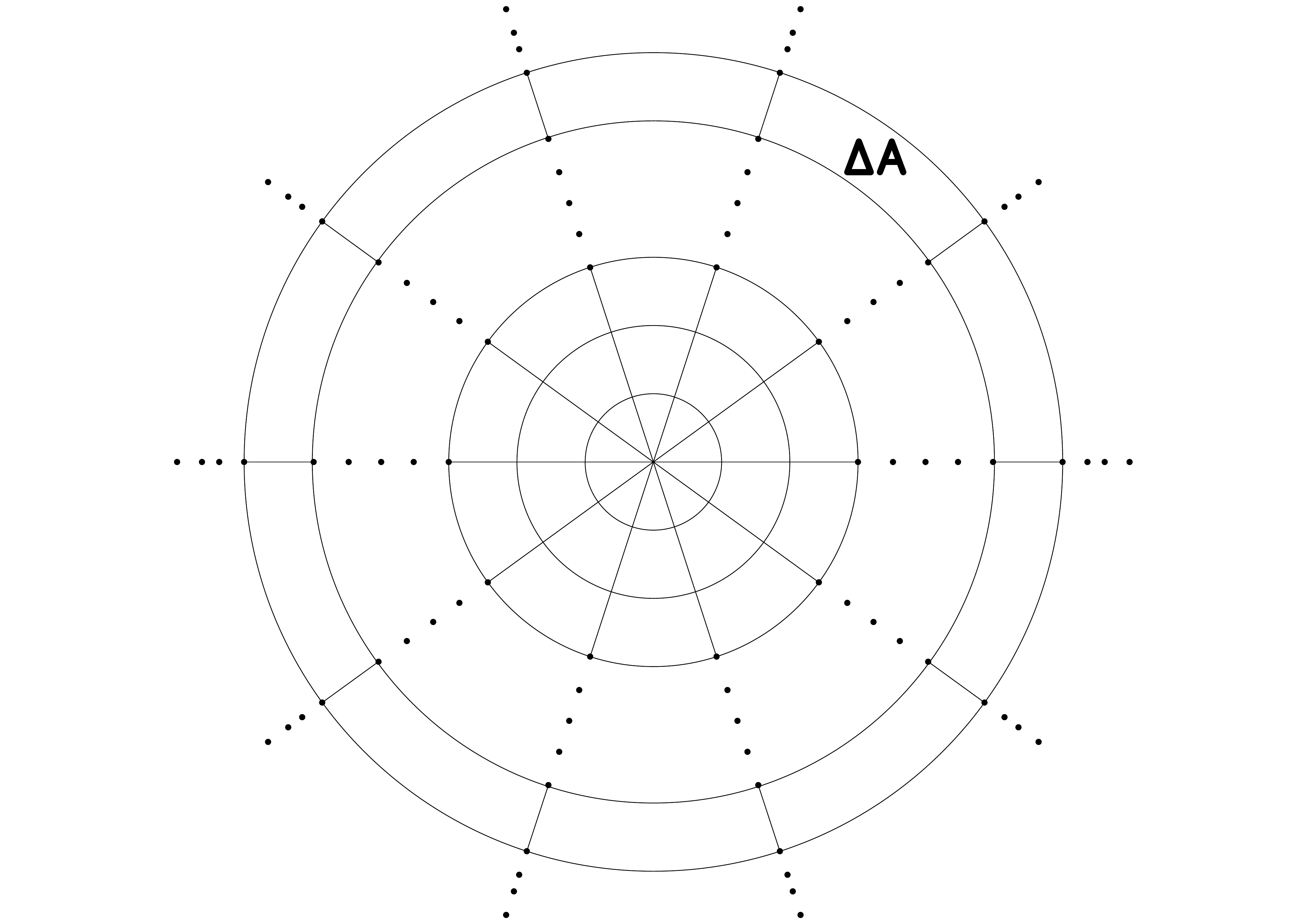}
\centerline{\includegraphics[width=3in]{final.pdf}}
\caption{$\Delta A$ between annulus for $L_0=10$}
\label{final.pdf}
\end{figure}

In this section, we shall present a modified grid, which prevents the angular expansion. For the dilations, we choose   the scale $ 0< \lambda_0 < 1$.
For rotation consider the matrix

$$ \mathcal{R}_{\lambda,\theta} =  \begin{pmatrix}
 \lambda  \cos(\theta) & - \lambda \sin(\theta)\\
  \lambda \sin(\theta) & \lambda \cos(\theta) \\
\end{pmatrix} .$$  For the translations, we have
 $ \vec{b}= \vec{b}_{tlm_0m_1}=t \mathcal{R}_{\lambda, \theta}(\vec{u}_{m_0m_1}), \mbox{ with}\vec{u}_{m_0m_1}= (m_0\beta_0, m_1\beta_1),~~ m_0, m_1 \in \mathbb{Z},~~\mbox{and for some}~~\beta_0, \beta_1\geq 0.$
Let $L>1$  be fixed and for $t \in \mathbb{Z}^{+}$, $\Delta\theta= \frac{2\pi}{tL},$  $\theta(t):= l\Delta\theta =\frac{l2\pi}{tL};~  l \in \mathbb{Z}_{t L}=\{0, 1, 2, 3, ..., tL-1\}$
 and $\lambda :=  \lambda_0, $ for some fixed $ 0<\lambda_0<1.$

 An estimate for the area $\Delta A $ in  Figure 1 (see also Figure 2)   is $ \Delta A \cong \Delta s  \Delta r =   r \Delta \theta \Delta r.$  If we choose $ 0< \lambda= \lambda_0<1,$ then in step  $t$ we have $r= t\lambda, \Delta\theta= \frac{2\pi}{tL}, \frac{(l-1)2\pi}{tL}\leq \theta=\theta(t)\leq \frac{l2\pi}{tL}$ and $ \Delta r=\lambda.$ So, $$\Delta A\cong t\lambda^2\Delta\theta = \lambda^2 2\pi/L.$$
In all steps we need the area $\Delta A $ to be less than a constant, say  $\eta.$  So in the first step we choose $\lambda$ so small and $ L$ so large  such that $\Delta  A = \frac{2\pi\lambda^2}{L} <\eta.$
We will choose the sampling grid $ \Lambda$ as follows
$$\Lambda = \Lambda (\lambda, L, \beta_0, \beta_1)=\{ (t\lambda, l\frac{ 2\pi}{tL}, \vec{b}_{tlm_0m_1}); (t, l,(m_0, m_1))\in \mathbb{Z}^{+} \times \mathbb{Z}_{tL} \times \mathbb{Z}^2 \}$$
 with wavelet coefficients $\{S_{tlm_0m_1}\vert (t, l, m_0, m_1)\in  \mathbb{Z}^{+} \times \mathbb{Z}_{tL} \times \mathbb{Z}^2 \}. \ \mbox{For} f\in L^2(\mathbb R^2),$
\begin{align}
S_{tlm_0m_1}= S(\vec{b}_m, t\lambda, l
 \Delta\theta )
 \nonumber
  = < \psi_{ \mathcal{R}_{\lambda, \theta}, \vec{b}_m}|f>
  \nonumber
 =\int_{\mathbb{R}^2}{d^2\vec{x}\overline{\psi_{ \mathcal{R}_{\lambda, \theta}, \vec{b}_m}(\vec{x})}f(\vec{x})}	
\end{align}
As usual the Fourier transform of $ \mathcal{R}_{\lambda, \theta}-$dilated and $\vec{b_m}-$ translated function $\psi$ is
$$ \widehat{\psi}_{ \mathcal{R}_{\lambda, \theta},\vec{b_m}}(\vec{x})=\frac{1}{2\pi}\int_{\mathbb{R}^2}e^{-i\vec{x} \dot{.} \vec{y}} \frac{1}{\lambda^2}\psi( \mathcal{R}^{-1}_{\lambda\theta}(\vec{y}-\vec{b_m}))d\vec{y},$$ using a change of variable we have
$$ \widehat{\psi}_{ \mathcal{R}_{\lambda, \theta},\vec{b_m}}(\vec{x})
=\frac{\lambda^2}{2\pi}e^{-i\vec{x}\dot{.}\vec{b_m}}\int_{\mathbb{R}^2}e^{-i \mathcal{R}_{\lambda, \theta}^T\vec{x}\dot{.} \vec{z}}\psi(\vec{z})d\vec{z}=\frac{\lambda^2}{2\pi}e^{-i\vec{x}\dot{.}\vec{b_m}}\widehat{\psi}( \mathcal{R}_{\lambda, \theta}^T \vec{x}). $$
Our task is now to find conditions on the grid  $\Lambda$  as defined above such that  $$\{\psi_{tlm}: (t, l, (m_0, m_1) )\in  \mathbb{Z}^{+} \times \mathbb{Z}_{tL}  \times \mathbb{Z}^2 \}$$
is a frame.
\begin{theorem}
Let $\psi$ be the wavelet in (\ref{2dimrep}) satisfying the following conditions
\begin{enumerate}
\item[(a)]$\displaystyle s(\lambda, L, \psi)=\mbox{essinf}_{\vec{k}\in \mathbb{R}^2} \sum_{t \in \mathbb{Z}^{+}}\sum_{l\in \mathbb{Z}_{tL}} |\widehat{\psi}(\mathcal{R}_{\lambda, \theta}^T\vec{k})|^2>0.$
\item[(b)]$\displaystyle S(\lambda, L, \psi)=\sup_{\vec{k}\in \mathbb{R}^2} \sum_{t \in \mathbb{Z}^{+}}\sum_{l\in \mathbb{Z}_{tL}} |\widehat{\psi}(\mathcal{R}_{\lambda, \theta}^T\vec{k})|^2 < \infty.$
\item[(c)]$\displaystyle\sup_{\vec{b}\in \mathbb{R}^2}(1+|\vec{b}|)^{1+\epsilon}\alpha(\vec{b})<\infty,~~\mbox {where} $\\
\item[(d)]$\displaystyle\alpha(\vec{b})= \sup_{\vec{k}\in \mathbb{R}^2} \sum_{t\in \mathbb{Z}^{+}}\sum_{l\in \mathbb{Z}_{tL}}| \hat{\psi}(\mathcal{R}_{\lambda, \theta}^T\vec{k}+\vec{b})|| \hat{\psi}(\mathcal{R}_{\lambda, \theta}^T\vec{k})|$
\end{enumerate}
Then there exist constants $\beta_0^c, \beta_1^c >0 $ such that
\begin{enumerate}
\item[(1)] $\forall \beta_0 \in (0, \beta_0^c), \beta_1 \in (0, \beta_1^c)$, the    family $\psi_{tlm_0m_1}$ associated to $(\lambda, L, \beta_0, \beta_1)$  is a frame of $L^2(\mathbb{R}^2, d^2\vec{x}).$
\item[(2)] $\forall \delta>0$, there exist $\beta_0 \in (\beta_0^c, \beta_0^c + \delta), \beta_1 \in (\beta_1^c, \beta_1^c + \delta),  $  such that the family  $\psi_{tlm_0m_1}$ associated to $(\lambda, L, \beta_0, \beta_1)$
is not a frame of $L^2(\mathbb{R}^2, d^2\vec{x}).$
\end{enumerate}
\end{theorem}
\begin{proof}Proof follows similar to the proof of Theorem (\ref{Qd}).
\end{proof}
%%%%%%%%%%%%%%%%%%%%%%%%%%%%%%%%%%%%%%%%%%%%%%%%%%%%%%%%%%%%%%%%%%%%%%%%%%%%%%%%%%%%%%%%%%%
\section{Quaternionic continuous wavelet transforms and its discretization}
In order to understand the discretization better and for the sake of completeness, we shall quote some results from \cite{AT}. Further, it should be declared that the discretization procedure described in \cite{Van} acts as a basis for the discretization of quaternionic continuous wavelet transforms.
\subsection{The quaternionic wavelet group}
Consider the action of $\quat^*$ on $\quat$ by left (or right) quaternionic (in our
representation, matrix) multiplication. It is clear that there are only two orbits
under this action,
$\{\bfrako\}$ (the zero quaternion) and $\quat^*$. Furthermore, this latter orbit is
{\em open and free\/} in the usual sense (see, e.g. \cite{Ali}, for a definition). Let
$$ \bfraka = \begin{pmatrix} w_1 & -\overline{w}_2\\ w_2 & \overline{w}_1 \end{pmatrix}
\in \quat^* \quad \text{and} \quad
  \bfrakx = \begin{pmatrix} z_1 & -\overline{z}_2\\ z_2 & \overline{z}_1 \end{pmatrix}
  \in \quat .$$
Then under left action
\be \bfrakx \longmapsto \bfrakx^\prime = \bfraka \bfrakx
  = \begin{pmatrix} w_1 z_1 - \overline{w}_2 z_2 & -
  \overline{w}_2\overline{z}_1 -w_1\overline{z}_2 \\ w_2 z_1 + \overline{w}_1 z_2 &
     \overline{w}_1 \overline{z}_1  - w_2 \overline{z}_2\end{pmatrix} .
\label{left-act}
\en
Note that this is compatible with  the action
$$ \begin{pmatrix} z_1 \\z_2\end{pmatrix} \longmapsto \begin{pmatrix} w_1 & -\overline{w}_2\\ w_2 & \overline{w}_1 \end{pmatrix} \begin{pmatrix} z_1 \\z_2\end{pmatrix}, $$
of the affine $SU(2)$ group (i.e., the group $\mathbb R^{>0}\times SU(2)$) on $\mathbb C^2$.
We write $w_1 =a_0 + ia_3, \; w_2 = a_2 + ia_1$ and $z_1 = x_0 + ix_3, \; z_2 =
x_2 + ix_1$ and consider $\bfrakx$ as the vector
\be
  \bx = \begin{pmatrix} x_0 \\x_3 \\ x_2 \\x_1 \end{pmatrix} \in \mathbb R^4 .
\label{vec-rep}
\en
On this vector, the left action (\ref{left-act}) is easily seen to lead to the matrix
left action
\be
  \bx \longmapsto \bx^\prime = A\bx = \begin{pmatrix} a_0 & -a_3 & -a_2 & -a_1 \\
                                                       a_3 & a_0 & a_1 & -a_2 \\
                                                       a_2 & -a_1 & a_0 & a_3 \\
                                                       a_1 & a_2 & -a_3 & a_0 \end{pmatrix}
                                                        \begin{pmatrix} x_0 \\x_3 \\ x_2 \\x_1 \end{pmatrix}
 = \begin{pmatrix} A_1 & -A_2^T \\ A_2 & A_1^T\end{pmatrix}
            \begin{pmatrix} \bx_1 \\ \bx_2 \end{pmatrix},
 \label{mat-left-act}
\en
on $\mathbb R^4$, where
$$
 A_1 = \begin{pmatrix} a_0 & -a_3 \\a_3 & a_0 \end{pmatrix}, \quad
  A_2 = \begin{pmatrix} a_2 & -a_1 \\a_1 & a_2 \end{pmatrix}, \quad
  \bx_1 = \begin{pmatrix} x_0 \\ x_3 \end{pmatrix}, \quad
  \bx_2 = \begin{pmatrix} x_2 \\ x_1 \end{pmatrix}. $$
 The matrices $A_1$ and $A_2$ are rotation-dilation matrices, and may be written in the form
 \be
    A_1 = \lambda_1 \begin{pmatrix} \cos\theta_1 & - \sin\theta_1\\
                                    \sin\theta_1 &  \cos\theta_1 \end{pmatrix}
    = \lambda_1 R(\theta_1), \qquad
    A_2 = \lambda_2 \begin{pmatrix} \cos\theta_2 & - \sin\theta_2 \\
                                    \sin\theta_2 &  \cos\theta_2 \end{pmatrix}
    = \lambda_2 R(\theta_2)
 \label{rot-dil-mat}
 \en
 where
 \be
  \theta_1 = \tan^{-1}\left(\frac {a_3}{a_0}\right), \;
  \theta_2 = \tan^{-1}\left(\frac {a_1}{a_2}\right), \;
    \lambda_1 = \sqrt{a_0^2 + a_3^2}, \;  \lambda_2 = \sqrt{a_1^2 + a_2^2} \; \text{and} \;
  \lambda_1^2 + \lambda_2^2 \neq 0
 \label{rot-dil-cond}
 \en
 and $R(\theta)$ is the $2\times 2$ rotation matrix
 \be
 R(\theta) = \begin{pmatrix} \cos\theta & -\sin\theta \\ \sin\theta & \cos\theta
   \end{pmatrix} .
 \label{rot-mat}
 \en
 Note that
 $$ A^T A = A A^T  = \vert\bfraka\vert^2 \mathbb I_4\; \quad \text{and} \quad
 \text{det}[A] = \vert\bfraka\vert^4 .$$

 From the above it is clear that when $\quat$ is identified with $\mathbb R^4$, the action of $\quat^*$ on $\quat$ is that of two
 two-dimensional rotation-dilation groups (rotations of the two-dimensional plane together
 with radial dilations, where at least one of the dilations is non-zero) acting on $\mathbb R^4$. Consequently, we shall consider elements in $\quat^*$  as $4\times 4$ real matrices of the type $A$ in (\ref{mat-left-act}):
\be
 A_{\lambda,\theta} = \begin{pmatrix} \lambda_1 R(\theta_1) & -\lambda_2 R(-\theta_2)\\
                        \lambda_2 R(\theta_2) & \lambda_1 R(-\theta_1), \end{pmatrix}, \qquad
\text{det}[A] = \vert\bfraka\vert^4  = [\lambda_1^2 + \lambda_2^2]^2\neq 0.
\label{second-mat-rep}
\en
The matrix $A$ has the inverse
$$
  A^{-1} = \frac 1{\lambda_1^2 + \lambda_2^2} \begin{pmatrix} \lambda_1 R(-\theta_1) & \lambda_2 R(-\theta_2)\\
                       - \lambda_2 R(\theta_2) & \lambda_1 R(\theta_1), \end{pmatrix} . $$

%%%%%%%%%%%%%%%%%%%%%%%%%%%%%%%%%%%%%%%%%%%%
\subsection{A little bit of all three affine groups}\label{sec-quat-aff-grp}
 Let us look at the three affine groups, $G^{\mathbb R}_{\text{aff}}, G^{\mathbb C}_{\text{aff}}$
 and $G^{\mathbb H}_{\text{aff}}$, of the real line, the complex plane and the quaternions,
 respectively. These groups are defined as the semi-direct products
 $$
   G^{\mathbb R}_{\text{aff}} = \mathbb R \rtimes \mathbb R^* , \qquad
   G^{\mathbb C}_{\text{aff}} = \mathbb C \rtimes \mathbb C^* , \qquad
   G^{\mathbb H}_{\text{aff}} = \mathbb H \rtimes \mathbb H^* . $$
Let $\mathbb K$ denote any one of the three fields $\mathbb R, \mathbb C$ or $\mathbb H$ and write $G^{\mathbb K}_{\text{aff}} = \mathbb K \rtimes \mathbb K^*$. A generic element in
$G^{\mathbb K}_{\text{aff}}$ can be written as
$$ g = (b, a) = \begin{pmatrix} a & b \\ 0 & 1 \end{pmatrix}, \quad a \in \mathbb K^* , \;\;
    b \in \mathbb K .$$
Of these, $G^{\mathbb R}_{\text{aff}}$ is the {\em one-dimensional wavelet group} and
$G^{\mathbb C}_{\text{aff}}$, which is isomorphic to the similitude group of the plane (translations, rotations and dilations of the 2-dimensional plane), is the {\em two-dimensional
wavelet group\/.} By analogy  the quaternionic affine group
$G^{\mathbb H}_{\text{aff}}$ is called the {\em quaternionic wavelet group\/} \cite{AT}.

From the general theory of semi-direct products of the type $\mathbb R^n \rtimes H$, where
$H$ is a subgroup of $GL(n, \mathbb R)$, and which has open free orbits in the dual of
$\mathbb R^n$, (see, for example, \cite{Ali}, Chapter 8), it is known that
$G^{\mathbb H}_{\text{aff}}$ has exactly one unitary irreducible representation
on a complex Hilbert space and
moreover, this representation is square-integrable.  Let us see the construction of this
representation (in a Hilbert space over the complexes). Consider the Hilbert space
$\h_{\mathbb C} = L^2_{\mathbb C} (\mathbb R^4, d\bx )$ and define on it the
representation $G^{\mathbb H}_{\text{aff}} \ni (\bb, A) \longmapsto U_{\mathbb C}(\bb , A)$,
\be
   (U_{\mathbb C}(\bb , A) f)(\bx ) = \frac 1{(\text{det}[A])^{\frac 12}}
   f( A^{-1} (\bx - \bb)) ,
   \qquad f \in \h_{\mathbb C} .
\label{comp-quat-rep}
\en
This representation is unitary, irreducible and square-integrable. From the
general theory \cite{Ali}, the {\em Duflo-Moore operator $C$}  is given in the
Fourier domain as the multiplication operator
\be
   (\widehat{C}\widehat{f})(\bk ) = \mathcal C (\bk )
                     \widehat{f} (\bk), \quad \text{where} \quad
   \mathcal C (\bk ) = \left[\frac {2\pi}{\Vert \bk\Vert}  \right]^2 .
\label{duflo-moore-op}
\en
A vector $f\in \h_{\mathbb C}$ is {\em admissible} if it is in the domain of $C$
i.e., if its Fourier transform $\widehat{f}$ satisfies
$$  (2\pi)^4\int_{\mathbb R^4} \frac {\vert \widehat{f}(\bk )\vert^2}
            {\Vert \bk \Vert^4} \; d\bk < \infty. $$

The above representation could also be realized on the Hilbert space
$\kc = L^2_{\mathbb C}(\quat , d\bfrakx )$ over the quaternions by simply transcribing
Eq. (\ref{comp-quat-rep}) into this framework.
Thus,  define the
representation $G^{\mathbb H}_{\text{aff}} \ni (\bfrakb, \bfraka) \longmapsto
U_{\mathbb C}(\bfrakb , \bfraka)$,
\be
   (U_{\mathbb C}(\bfrakb , \bfraka) f)(\bfrakx ) = \frac 1{\text{det}[\bfraka]}
   f( \bfraka^{-1} (\bfrakx - \bfrakb)) ,
   \qquad f \in \kc ,
\label{comp-quat-quat-rep}
\en
The {\em Duflo-Moore operator $C$}  is given in the
Fourier domain as the multiplication operator
\be
   (\widehat{C}\widehat{f})(\bfrakk ) = \mathcal C (\bfrakk )
                     \widehat{f} (\bfrakk), \quad \text{where} \quad
   \mathcal C (\bfrakk ) = \left[\frac {2\pi}{\vert \bfrakk\vert}  \right]^2 .
\label{quat-duflo-moore-op}
\en
The admissibility condition is now
$$  (2\pi)^4\int_{\mathbb R^4} \frac {\vert \widehat{f}(\bfrakk )\vert^2}
            {\vert \bfrakk \vert^4} \; d\bfrakk < \infty. $$
For further details we refer the reader to \cite{AT}.
%%%%%%%%%%%%%%%%%%%%%%%%%%%%%%%%%%%%%%%%%%%%%%%%%%%%%%%%%%%%%%%%%%%%%%%%%%%%%%%%%%%%%%%%%%%%%%%%%%%%%%%%%%%%%%%%%%%%%%%%%
\subsection{ Discretization of quaternionic wavelet transform }
For the dilations, we choose   the scale $ 0< \lambda_1=\lambda_{01} < 1$.  For rotations consider the matrix  $$ \mathcal{R}_{\lambda_1, \theta_1} =  \begin{pmatrix}
 \lambda_1  \cos( \theta_1) & - \lambda_1 \sin(\theta_1)\\
  \lambda_1 \sin(\theta_1) & \lambda_1 \cos(\theta_1) \\
\end{pmatrix} .$$\\ Let $L_1>1$  be fixed and for $t \in \mathbb{Z}^{+}, \ \ \Delta\theta_1= \frac{2\pi}{tL_1}, \ \  \theta_1(t):= l\Delta\theta_1 =\frac{l2\pi}{tL}, \ \  l \in \mathbb{Z}_{t L_1}=\{0, 1, 2, 3, ..., tL_1-1\}$
 and $\lambda_1 := t \lambda_{01}, $ for some fixed $ 1>\lambda_{01}>0.$ For the translations, we have
 $ \vec{b}= \vec{b}_{tlm_0m_1}=t \mathcal{R}_{\lambda_1, \theta_1}(\vec{u}_{m_0m_1}),$ with $\vec{u}_{m_0m_1}=(m_0\beta_0, m_1\beta_1), \ m_0, m_1 \in \mathbb{Z},$ and for some $\beta_0, \beta_1\geq 0.$\\

%%%%%%%%%%%%%%%%%%%%%%%%%%%%
\begin{figure}[h]
\begin{center}
\subfigure[$L_1=2 \ and  \ t=1, 2, 3, 4, 5. $]{
\includegraphics[width=6cm]{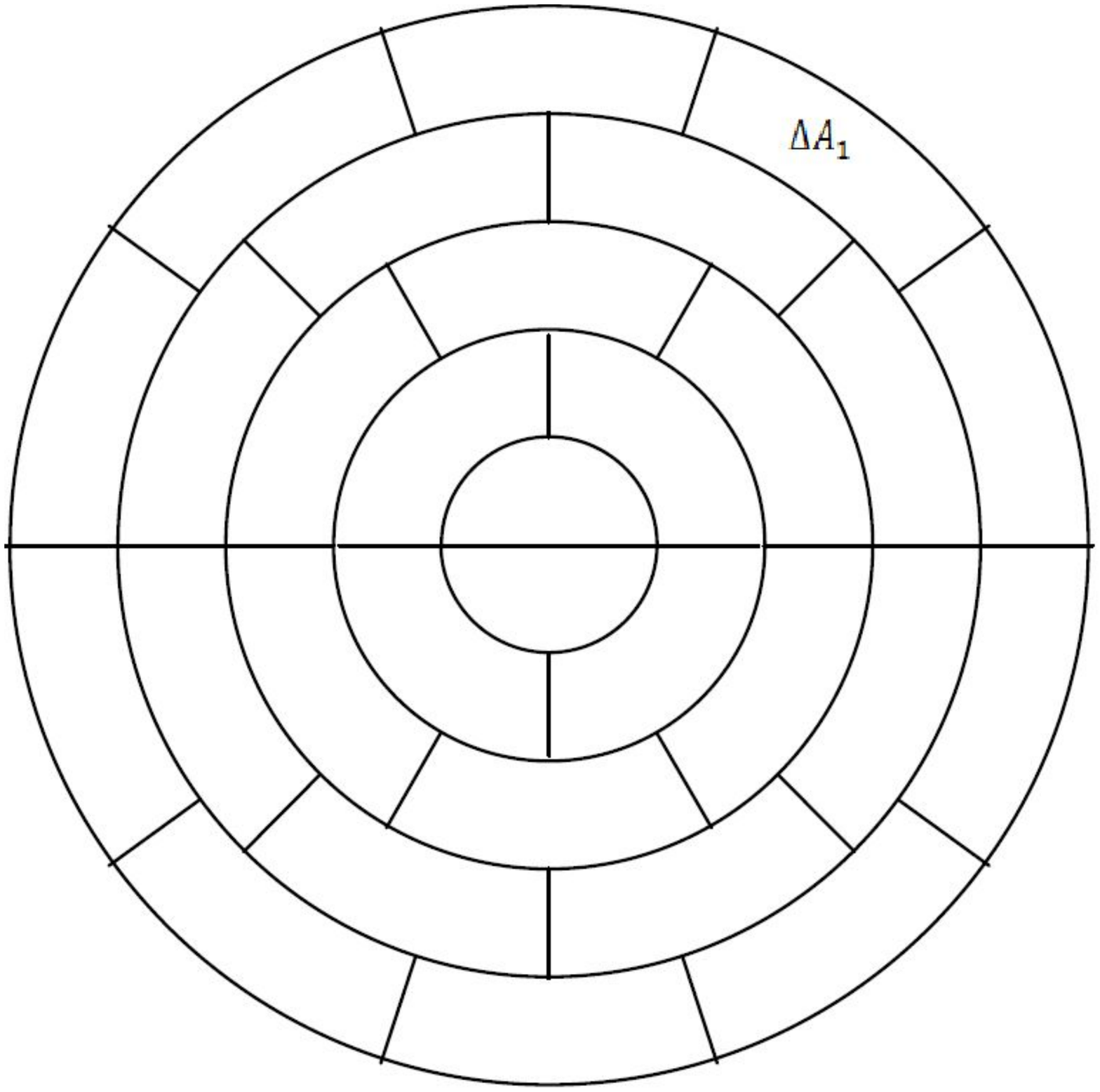}
}
\subfigure[$L_1=3 \ and  \ t=1, 2, 3, 4, 5. $]{
\includegraphics[width=6cm]{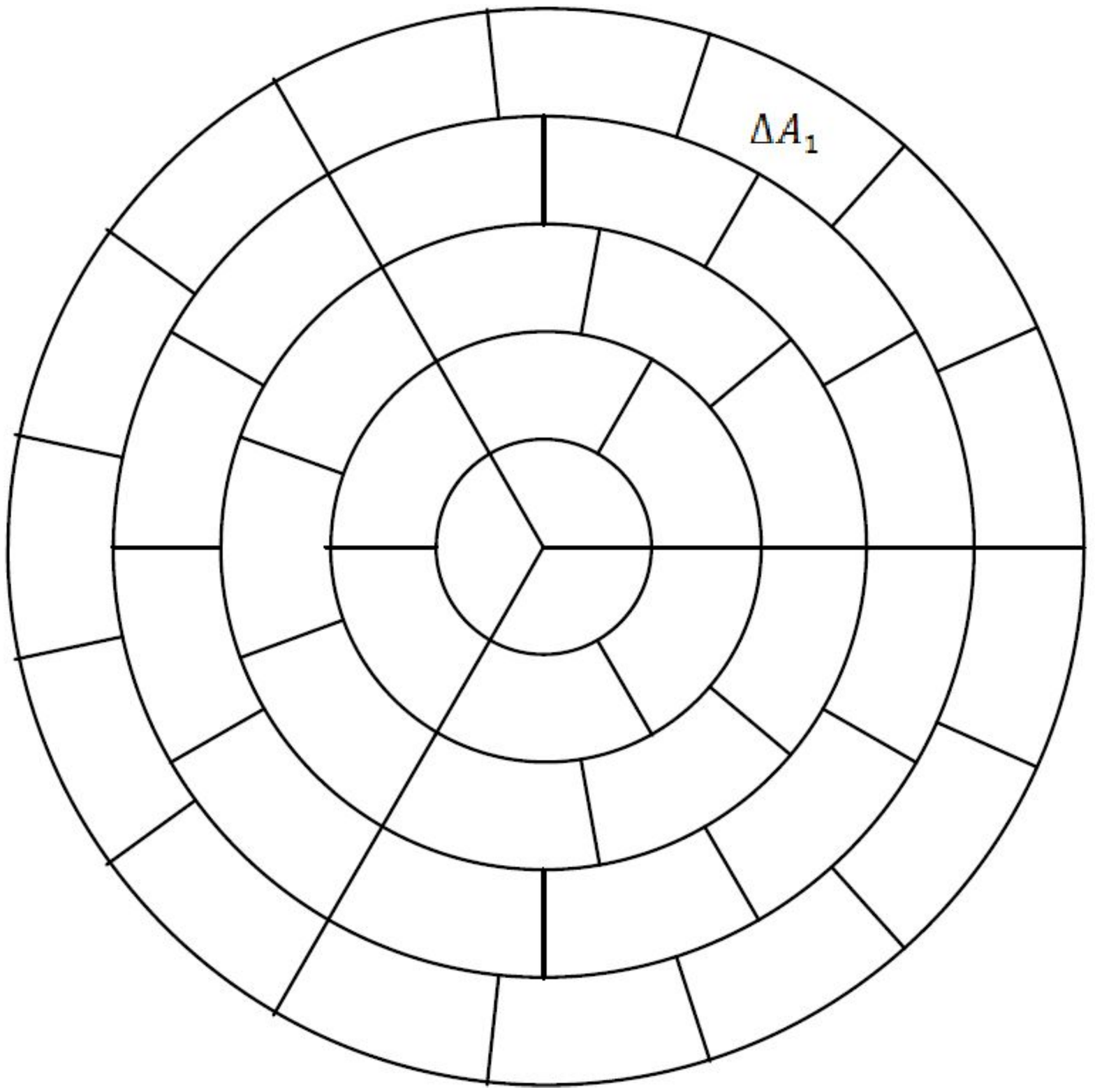}
}
\caption{$\Delta A_1$ between annulus. }
\label{1.pdf}
\end{center}
\end{figure}
%%%%%%%%%%%%%%%%%%%%%%%%%%
%\begin{figure}[h]
%\centering
%\include{1.pdf}
%\centerline{\includegraphics[width=3in]{1.pdf}}
%\caption{$\Delta A_1$ between annulus for $L_1=3 \ and  \ t=1, 2, 3, 4, 5. $}
%\label{1.pdf}
%\end{figure}

To estimate the area of $\Delta A_1 $ in  Figure 2 we choose $ 0< \lambda_1= \lambda_{01}<1,$ then in step  $t$ we have $r= t\lambda_1, \Delta\theta_1= \frac{2\pi}{tL_1},\frac{(l-1)2\pi}{tL_1}\leq \theta=\theta(t)\leq \frac{l2\pi}{tL_1}$ and $ \Delta r=\lambda_1.$
 In fact, for fixed $L_1>1$   in step $ t $  we divide the annulus, between two circles of radius $ (t-1)\lambda_1 $ and  $ t\lambda_1,$ into $tL_1$ equal parts and the exact value of the area of $\Delta A_1 $ is equal to $\frac{\pi (2t-1)\lambda_1^2}{tL_1}.$  So $\Delta A_1 = \frac{\pi \lambda_1^2}{L_1}$  for  $t=1,$  and $\frac{\pi \lambda_1^2}{L_1}<\Delta A < \frac{3\pi \lambda_1^2}{2L_1}$ for $ t> 1$ and $\Delta A \rightarrow \frac{\pi \lambda_1^2}{L_1}$ as  $ t \rightarrow +\infty.$    In all steps we need the area $\Delta A_1 $ to be less than a constant, say  $\eta_1.$ In the first step we choose $\lambda_1$ so small and $ L_1$ so large  such that $\Delta  A_1 = \frac{3\pi\lambda_1^2}{2L_1} <\eta_1.$
 The above argument  for $\lambda_{01}, \lambda_1, \theta_1, L_1, t, l, m_0, m_1, \beta_0, \beta_1, \Delta A_1, \Delta\theta_1 \  \mbox{and} \  \eta_1$,   works also for $\lambda_{02}, \lambda_2, \theta_2, L_2, j, k, m_2, m_3, \beta_2, \beta_3, \Delta A_2, \Delta\theta_2  \ \mbox{and} \  \eta_2,$  as well.
 Using a similar argument as in \cite{Van} and the above paragraphs we obtain the following matrix
$$A_{\lambda\theta}=  \begin{pmatrix}
   \mathcal{R}_{\lambda_1, \theta_1} & -\mathcal{R}_{\lambda_2,  -\theta_2} \\ \mathcal{R}_{\lambda_2, \theta_2} &  \mathcal{R}_{\lambda_1,  -\theta_1} \\
\end{pmatrix}. $$
 For discretization we will choose the sampling grid $ \Lambda$ as follows:
\begin{eqnarray*}
\Lambda &=& \Lambda (\lambda_1, \lambda_2, L_1, L_2, \beta_0, \beta_1, \beta_2, \beta_3)\\
&=&\{ (t\lambda_1, j\lambda_2, l\frac{ 2\pi}{tL_1}, k\frac{ 2\pi}{jL_2}, \vec{b}_{tjlkm}); ((t, j), l,k,m)\in \mathbb{Z}^2 \times \mathbb{Z}_{tL_1} \times \mathbb{Z}_{jL_2} \times \mathbb{Z}^4 \},
\end{eqnarray*}
where $ \mathbb{Z}_{aL}=\{0, 1, 2, 3,,..., aL-1\}$,
  $\vec{b}_m=\vec{b}_{tjlkm_0m_1m_2m_3}.$ With wavelet coefficients $$\{S_{tjlkm_0m_1m_2m_3}: (t, j, l, k, m_0, m_1, m_2, m_3 )\in  \mathbb{Z}^2 \times \mathbb{Z}_{tL_1} \times \mathbb{Z}_{jL_2} \times \mathbb{Z}^4 \}$$
\begin{eqnarray*}
S_{tjlkm_0m_1m_2m_3}= S(\vec{b}_m, \lambda_1^{-t}, \lambda_2^{-j}, l
\theta_{01}, k \theta_{02} )
 = < \psi_{A_{\lambda\theta}, \vec{b}_m}|s>
& =\int_{\mathbb{R}^4}{d^4\vec{x}\overline{\psi_{A_{\lambda\theta}, \vec{b}_m}(\vec{x})}s(\vec{x})}.	
\end{eqnarray*}
As usual the Fourier transform of $A_{\lambda\theta}-$dilated and $\vec{b_m}-$ translated function $\psi$ is
$$ \widehat{\psi}_{A_{\lambda\theta},\vec{b_m}}(\vec{x})=\frac{1}{(2\pi)^2}\int_{\mathbb{R}^4}e^{-i\vec{x} \dot{.} \vec{y}} \frac{1}{\lambda_1^2+\lambda_2^2}\psi(A^{-1}_{\lambda\theta}(\vec{y}-\vec{b_m}))d\vec{y}.$$ Using  change of variables we have
$$ \widehat{\psi}_{A_{\lambda\theta},\vec{b_m}}(\vec{x})=\frac{\lambda^2_1+\lambda^2_2}{(2\pi)^2}e^{-i\vec{x}\dot{.}\vec{b_m}}\int_{\mathbb{R}^4}e^{-iA_{\lambda\theta}^T\vec{x}\dot{.} \vec{z}}\psi(\vec{z})d\vec{z}=\frac{\lambda_1^2+\lambda_2^2}{(2\pi)^2}e^{-i\vec{x}\dot{.}\vec{b_m}}\widehat{\psi}(A_{\lambda\theta}^T \vec{x}). $$
Our task is now to find conditions on the grid  $\Lambda$  as defined above such that  $$\{\psi_{tjlkm}: ((t, j), l, k, m )\in  \mathbb{Z}^2 \times \mathbb{Z}_{tL_1} \times \mathbb{Z}_{jL_2} \times \mathbb{Z}^4 \}$$
is a frame. We will use the following notation
$\mathcal{L}_{tj}=\mathbb{Z}_{tL_1}\times \mathbb{Z}_{jL_2}.$ The  proof of the following theorem is more or less an adaptation, to our new grid, of the proof given in \cite{Van}.
%%%%%%%%%%%%%%%%%%%%%%%%%%%%%%%%%%%%%%%%%%%%%%%%%%%%%%%%%%%%%%%%%%%%%%%%%%%%%%%%%%%%%%%%%%%%%%%%
\begin{theorem}\label{Qd}
Let $\psi$ be a quaternionic wavelet satisfying the following conditions:
\begin{enumerate}
\item[(a)] $\displaystyle s(\lambda_1, \lambda_2, L_1, L_2, \psi)=essinf_{\vec{k}\in \mathbb{R}^4} \sum_{t, j \in \mathbb{Z}}\sum_{(l, k)\in \mathcal{L}_{tj}} |\widehat{\psi}(A^T_{\lambda\theta}\vec{k})|^2>0.$
\item[(b)]$\displaystyle  S(\lambda_1, \lambda_2, L_1, L_2, \psi)=\sup_{\vec{k}\in \mathbb{R}^4} \sum_{t, j \in \mathbb{Z}}\sum_{(l, k)\in \mathcal{L}_{tj}} |\widehat{\psi}(A^T_{\lambda\theta}\vec{k})|^2 < \infty.$
\item[(c)]$\displaystyle  \sup_{\vec{b}\in \mathbb{R}^4}(1+|\vec{b}|)^{1+\epsilon}\alpha(\vec{b})<\infty,$  where\\
$\displaystyle \alpha(\vec{b})= \sup_{\vec{k}\in \mathbb{R}^4} \sum_{t, j\in \mathbb{Z}}\sum_{(l, k)\in \mathcal{L}_{tj}}| \hat{\psi}(A^T_{\lambda\theta}\vec{k}+\vec{b})|| \hat{\psi}(A^T_{\lambda\theta}\vec{k})|.$
\end{enumerate}
Then there exist constants $\beta_0^c, \beta_1^c, \beta_2^c, \beta_3^c >0 $ such that
\begin{enumerate}
 \item[(1)] $\displaystyle\forall \beta_0 \in (0, \beta_0^c), \beta_1 \in (0, \beta_1^c), \beta_2 \in (0, \beta_2^c), \beta_3 \in (0, \beta_3^c)$, the    family $\displaystyle\psi_{tjlkm_0m_1m_2m_3}$ associated with $\displaystyle(\lambda_1, \lambda_2, L_1, L_2, \beta_0, \beta_1, \beta_2, \beta_3)$  is a frame of $\displaystyle L^2(\mathbb{R}^4, d^4\vec{x}).$
\item[(2)] $\displaystyle\forall \delta>0$, there exist $\displaystyle\beta_0 \in (\beta_0^c, \beta_0^c + \delta), \beta_1 \in (\beta_1^c, \beta_1^c + \delta), \beta_2 \in (\beta_2^c, \beta_2^c + \delta), \beta_3 \in (\beta_3^c, \beta_3^c + \delta),  $  such that the family  $\displaystyle\psi_{tjlkm_0m_1m_2m_3}$ associated with $\displaystyle(\lambda_1, \lambda_2, L_1, L_2, \beta_0, \beta_1, \beta_2, \beta_3)$
is not a frame of $\displaystyle L^2(\mathbb{R}^4, d^4\vec{x}).$
\end{enumerate}
\end{theorem}
\begin{proof}
We are looking for conditions on eight parameters $\lambda_1, \lambda_2, L_1, L_2, \beta_0, \beta_1, \beta_2, \beta_3 $ for which there exist $0< A\leq B<\infty$ such that
 $$ A||f||^2\leq \sum_{(t,j)\in \mathbb{Z}^2} \sum_{(l, k)\in \mathcal{L}_{tj}}\sum_{m\in \mathbb{Z}^4}|<\psi_{A_{\lambda\theta}, \vec{b_m}}|f>|^2\leq B||f||^2.$$
 For the middle term in the above formula we have
\begin{eqnarray*}
K &=& \sum_{(t,j)\in \mathbb{Z}^2} \sum_{(l, k)\in \mathcal{L}_{tj}}\sum_{m\in \mathbb{Z}^4}|\left\langle \psi_{A_{\lambda\theta}, \vec{b_m}}  | f \right\rangle|^2 \\
&=& \sum_{(t,j)\in \mathbb{Z}^2} \sum_{(l, k)\in \mathcal{L}_{tj}}\sum_{m\in \mathbb{Z}^4}
\frac{(\lambda_1^2+\lambda_2^2)^2}{(2\pi)^4}
 \int_{\mathbb{R}^4}d^4\vec{k}
 \int_{\mathbb{R}^4}{d^4\vec{k}^{'}} {e^{ib_m.(\vec{k}-\vec{k}^{'} )}} \widehat{\psi}(A_{\lambda\theta}^T\vec{k}^{'}) \overline{\widehat{\psi}(A_{\lambda\theta}^T\vec{k})}  \overline{\hat{f}(\vec{k^{'}})}\hat{f}(\vec{k})
  \end{eqnarray*}
Using the Poisson formula , for $d=0, 1, 2, 3, $ we have
\begin{eqnarray*}
\sum_{m_d=-\infty}^{+\infty}{e^{i 2\pi m_d\frac{\beta_d}{2\pi}(k_d-k_d^{'})}}&=&\dfrac{2\pi}{\beta_d}\sum_{m_d=-\infty}^{+\infty} \frac{\beta_d}{2\pi} {e^{i 2\pi m_d\frac{\beta_d}{2\pi}(k_d-k_d^{'})}}
=\dfrac{2\pi}{\beta_d}\sum_{m_d=-\infty}^{+\infty} \delta(k_d-k_d^{'}-m_d\frac{2\pi}{\beta_d})\\
&=&\dfrac{2\pi}{\beta_d}\sum_{m_d=-\infty}^{+\infty} \delta(k_d-k_d^{'}- \tilde{b}_{m_d}),
\end{eqnarray*}
 where $\tilde{b}_{m_d} =m_d \dfrac{2\pi}{\beta_d}. $
 Since $\delta(\vec{k}-\vec{K^{'}}-m)=\prod_{d=0}^3 \delta(k_d-k^{'}_d-m)$, we have
\begin{align}
\nonumber
 \sum_{m\in\mathbb{Z}^4}{e^{i\vec{b}_m . (\vec{k}-\vec{k}^{'})}}
&=  \dfrac{16\pi^4}{\beta_0\beta_1\beta_2\beta_3}\sum_{m\in \mathbb{Z}^6} \delta (\vec{k}-\vec{k}^{'}-\vec{\tilde{b}}_m ),
\nonumber
\end{align}
where $\vec{\tilde{b}}_m =(m_0\frac{2\pi}{\beta_0}, m_1\frac{2\pi}{\beta_1}, m_2\frac{2\pi}{\beta_2}, m_3\frac{2\pi}{\beta_3}). $
Now we have
\nonumber
\begin{align}K  & =
 \dfrac{(\lambda_1^2+\lambda_2^2)^2}{\beta_0\beta_1\beta_2\beta_3} \sum_{(t,j)\in \mathbb{Z}^2} \sum_{(l, k)\in \mathcal{L}_{tj}}\sum_{m\in \mathbb{Z}^4}\int_{\mathbb{R}^4}d^4\vec{k} \ \hat{\psi}(A^T_{\lambda\theta}(\vec{k}-\vec{\tilde{b}}_m))\  \overline{\widehat{\psi}(A^T_{\lambda\theta}\vec{k})} \  \overline{\widehat{f}(A^T_{\lambda\theta}(\vec{k}-{\tilde{b}}_m))} \   \widehat{f}(A^T_{\lambda\theta}\vec{k}).
 \end{align}
 We will split the above sum as $K=P+Q$, where
 $P$  denotes the term with $ m=(m_0, m_1, m_2, m_3)=(0, 0, 0, 0)$ and $Q$ the rest. In the following we set $ \mathbb{Z}^4_*=\mathbb{Z}^4 \setminus (0, 0, 0, 0).$
\begin{eqnarray*}
K&=& P+Q
= \dfrac{(\lambda_1^2+\lambda_2^2)^2}{\beta_0\beta_1\beta_2\beta_3}\sum_{(t, j)\in\mathbb{Z}^2} \sum_{(l, k) \in \mathcal{L}_{tj}} \int_{\mathbb{R}^4}d^2\vec{k}|\hat{\psi}(A^T_{\lambda\theta}\vec{k}))|^2|\hat{f}(A^T_{\lambda\theta}\vec{k})|^2 \\
&+&
 \frac{(\lambda_1^2+\lambda_2^2)^2}{\beta_0\beta_1\beta_2\beta_3} \sum_{(t,j)\in \mathbb{Z}^2} \sum_{(l, k)\in \mathcal{L}_{tj}}\sum_{m\in \mathbb{Z}^4_{*}}\int_{\mathbb{R}^4}d^4\vec{k} \ \hat{\psi}(A^T_{\lambda\theta}(\vec{k}-\vec{\tilde{b}}_m))\  \overline{\widehat{\psi}(A^T_{\lambda\theta}\vec{k})} \  \overline{\widehat{f}(A^T_{\lambda\theta}(\vec{k}-{\tilde{b}}_m))} \   \widehat{f}(A^T_{\lambda\theta}\vec{k})
\end{eqnarray*}
Using our hypothesis we get an estimate for each term in the above summation. For the first term we get from $ (a) $ and $ (b) $ 
$$ \frac{(\lambda_1^2+\lambda_2^2)^2}{\beta_0\beta_1\beta_2\beta_3} s(\lambda_1, \lambda_2, L_1, L_2,\psi)||\hat{f}||^2 \leq |P| \leq \frac{(\lambda_1^2+\lambda_2^2)^2}{\beta_0\beta_1\beta_2\beta_3} S(\lambda_1, \lambda_2, L_1, L_2,\psi)||\hat{f}||^2, $$
where     $ s(\lambda_1, \lambda_2, L_1, L_2,\psi)$  and $ S(\lambda_1, \lambda_2, L_1, L_2,\psi)$ are defined in $(a)$ and $(b)$, respectively.
By using the Cauchy-Schwartz inequality in  the second term , we get
\begin{align}
\nonumber
& |Q|\leq \frac{(\lambda_1^2+\lambda_2^2)^2}{\beta_0\beta_1\beta_2\beta_3}\sum_{(t,j)\in \mathbb{Z}^2} \sum_{(l, k)\in \mathcal{L}_{tj}}\sum_{m\in \mathbb{Z}^4_{*}}\int_{\mathbb{R}^4}d^4\vec{k}|\widehat{\psi}(A^T_{\lambda\theta}(\vec{k}-\vec{\tilde{b}}_m))|   |\widehat{\psi}(A^T_{\lambda\theta}\vec{k})||\hat{f}(A^T_{\lambda\theta}\vec{k})| | \hat{f}(A^T_{\lambda\theta}(\vec{k}- \vec{\tilde{b}}_m))|\\
\nonumber
& \leq \frac{(\lambda_1^2+\lambda_2^2)^2}{\beta_0\beta_1\beta_2\beta_3}\sum_{(t,j)\in \mathbb{Z}^2} \sum_{(l, k)\in \mathcal{L}_{tj}}\sum_{m\in \mathbb{Z}^4_{*}}[\int_{\mathbb{R}^4}d^4\vec{k}|\widehat{\psi}(A^T_{\lambda\theta}(\vec{k}-\vec{\tilde{b}}_m))|  |\widehat{\psi}(A^T_{\lambda\theta}(\vec{k}))| |\hat{f}(A^T_{\lambda\theta}\vec{k})|^2]^{\frac{1}{2}}\\
 \nonumber
& \times [\int_{\mathbb{R}^4}d^4\vec{k}|\widehat{\psi}(A^T_{\lambda\theta}(\vec{k}-\vec{\tilde{b}}_m))|  |\widehat{\psi}(A^T_{\lambda\theta}(\vec{k}))| |\hat{f}(A^T_{\lambda\theta}(\vec{k}-\vec{\tilde{b}}_m))|^2]^{\frac{1}{2}}\\
& \leq \frac{(\lambda_1^2+\lambda_2^2)^2}{\beta_0\beta_1\beta_2\beta_3}\sum_{(t,j)\in \mathbb{Z}^2} \sum_{(l, k)\in \mathcal{L}_{tj}}\sum_{m\in \mathbb{Z}^4_{*}}[\int_{\mathbb{R}^4}d^4\vec{k}|\widehat{\psi}(A^T_{\lambda\theta}(\vec{k}-\vec{\tilde{b}}_m))|  |\widehat{\psi}(A^T_{\lambda\theta}(\vec{k}))| |\hat{f}(A^T_{\lambda\theta}\vec{k})|^2]^{\frac{1}{2}}\\
 \nonumber
& \times [\int_{\mathbb{R}^4}d^4\vec{k}|\widehat{\psi}(A^T_{\lambda\theta}\vec{k})|  |\widehat{\psi}(A^T_{\lambda\theta}(\vec{k}+\vec{\tilde{b}}_m))| |\hat{f}(A^T_{\lambda\theta}\vec{k})|^2]^{\frac{1}{2}}\\
\end{align}
The Cauchy-Schwartz inequality for the infinite series on $l, k, t$ and $ j$ yields the following result
\begin{align}
\nonumber
& |Q|\leq \frac{(\lambda_1^2+\lambda_2^2)^2}{\beta_0\beta_1\beta_2\beta_3}\sum_{m\in \mathbb{Z}^4_{*}}[\int_{\mathbb{R}^4}d^4\vec{k}\{\sum_{(t,j)\in \mathbb{Z}^2}\sum_{(l, k)\in \mathcal{L}_{tj}}|\widehat{\psi}(A^T_{\lambda\theta}(\vec{k}-\vec{\tilde{b}}_m))| |\widehat{\psi}(A^T_{\lambda\theta}(\vec{k}))|\} |\hat{f}(A^T_{\lambda\theta}\vec{k})|^2]^{\frac{1}{2}}\\
 \nonumber
& \times [\int_{\mathbb{R}^4}d^4\vec{k}\{\sum_{(t,j)\in \mathbb{Z}^2}\sum_{(l, k)\in \mathcal{L}_{tj}}|\widehat{\psi}(A^T_{\lambda\theta}\vec{k})|  |\widehat{\psi}(A^T_{\lambda\theta}(\vec{k}+\vec{\tilde{b}}_m))|\}|\hat{f}(A^T_{\lambda\theta}\vec{k})|^2]^{\frac{1}{2}}\\
\end{align}
Using $(c)$  we have
$$|Q|\leq \frac{(\lambda_1^2+\lambda_2^2)^2}{\beta_0\beta_1\beta_2\beta_3}\{ \sum_{m\in \mathbb{Z}_{*}^4} \alpha(\vec{\tilde{b}}_m). \alpha(-\vec{\tilde{b}}_m)\}||\widehat{f}||^2.$$
To get the complete upper and lower bounds for $K$, we define
$$ E(\lambda_1, \lambda_2, \beta_0, \beta_1, \beta_2, \beta_3)= \sum_{m\in \mathbb{Z}_{*}^4} \alpha(\vec{\tilde{b}}_m). \alpha(-\vec{\tilde{b}}_m).$$
Using the inequality $|P|-|Q|\leq K \leq |P|+|Q|$
we obtain a lower bound for the left-hand side
$$ \frac{(\lambda_1^2+\lambda_2^2)^2}{\beta_0\beta_1\beta_2\beta_3}\{ s(\lambda_1, \lambda_2, L_1, L_2, \psi)- E(\lambda_1, \lambda_2, \beta_0, \beta_1, \beta_2, \beta_3)\}||\widehat{f}||^2\leq |P|-|Q|, $$
and an upper bound for right-hand side
$$ |P|-|Q|\leq\frac{(\lambda_1^2+\lambda_2^2)^2}{\beta_0\beta_1\beta_2\beta_3}\{ s(\lambda_1, \lambda_2, L_1, L_2, \psi)- E(\lambda_1, \lambda_2, \beta_0, \beta_1, \beta_2, \beta_3)\}||\widehat{f}||^2 . $$
\end{proof}
Using the discretization of the continuous wavelet transform of $\kc$  obtained in the above theorem, Theorem \ref{Qd} and the frame lifting described in section 4, we can write a discretized version of the quaternionic continuous wavelet transform, considered in this manuscript, of $\hquat$ as follows.
\begin{theorem}For $(t,j)\in\Z^2$, $(l,k)\in\mathcal{L}_{tj}$ and $m\in\Z^4$, according to theorem \ref{FLT},
\be
  \bmid \bPhi_{tjlkm}\bket = \begin{pmatrix} \vert \psi_{A_{\lambda\theta}, \vec{b_m}} \rangle&  0\\
                                       0 &  \vert \overline{\psi_{A_{\lambda\theta}, \vec{b_m}}} \rangle
                                        \end{pmatrix},
\label{quat-onbf1}
\en
where $\vert \psi_{A_{\lambda\theta}, \vec{b_m}} \rangle$ is as in Theorem \ref{Qd}, is a frame for $\hquat$ .
\end{theorem}
\section{conclusion}
We have introduced a general scheme for lifting a basis and a frame to a quaternionic Hilbert space from its real and complex counterparts. However, one should be aware that many more bases and frames exist in the quaternionic Hilbert spaces which are not of the type we discussed in this article. Our main goal, in this article, was to construct a discrete wavelet on a quaternion valued square integrable function space. In doing so, we have refined the existing theory on the complex Hilbert spaces. From the wavelet point of view, we do not consider that we have completed the task of constructing  wavelets in the quaternionic Hilbert spaces. Duplicating the discretization of the real wavelet transform to its complex or quaternion counterparts is still an open problems, and a workable multi-resolution analysis along the lines of the real case is still open. On the engineering side, there were several attempts to duplicate the real theory to complexes but none of them is as sophisticated as the real case (for the real wavelet analysis see, for example \cite{Dau, Ali} and for the $2D$ wavelets, see for example \cite{Van, Ben, Nic} and  many references therein). As a final word: in discretizating the continuous wavelet transform in quaternionic Hilbert spaces, we moved a step further in applying the initially defined continuous wavelet transforms to stereophonic or stereoscopic signals.
We would like to emphasize a fact that: the difficulty in applying the so-obtained discrete frame arises from computing their dual for real life data (even from mathematical point of view). This manuscript may provide motivation to someone to come with a way to write an algorithm to handle the so-called stereotype  data in the setting provided in this note.

\end{document}